\documentclass[11pt]{article}

\usepackage{amssymb}
\usepackage{graphicx}

\usepackage{tikz}  %needed for \textcolor as well as tikz
\usetikzlibrary{arrows}
\pgfdeclarelayer{bg}    % declare background layer
\pgfsetlayers{bg,main} % set the order of the layers (main is the standard layer)

\usepackage[colorinlistoftodos]{todonotes}

\usepackage{url}
\urldef{\mailsa}\path|{paul.goldberg, francisco.marmolejocossio}@cs.ox.ac.uk|

\usepackage{algorithm}
\usepackage[noend]{algorithmic}

\newcommand{\algorithmicoutput}{\textbf{Output:}}
\newcommand{\OUTPUT}{\item[\algorithmicoutput]}

\newcommand{\algorithmicinput}{\textbf{Input:}}
\newcommand{\INPUT}{\item[\algorithmicinput]}

\newcommand{\hide}[1]{}

\usepackage{fullpage}

\usepackage[]{hyperref}

\usepackage{amsmath, amssymb, amsthm}
\usepackage{mathtools}
\usepackage{verbatim}
\usepackage{color}
\usepackage{xcolor}

\usepackage{multicol}

\definecolor{DarkGreen}{rgb}{0.1,0.5,0.1}
\definecolor{DarkRed}{rgb}{0.5,0.1,0.1}
\definecolor{DarkBlue}{rgb}{0.1,0.1,0.5}

\usepackage{xspace}
\usepackage{cleveref}
\usepackage{thmtools, thm-restate}
\usepackage[numbers,sort&compress]{natbib}

\usepackage[mathscr]{euscript}
 \let\mathscr\relax% just so we can load this and rsfs
\usepackage[scr]{rsfso}

\def\epsilon{\varepsilon}

\renewcommand{\hat}{\widehat}
\renewcommand{\bar}{\overline}

\DeclareMathOperator*{\argmin}{\mathrm{argmin}}

\newtheorem*{theorem*}{Theorem}
\newtheorem*{observation*}{Observation}
\declaretheorem[
  name=Theorem,
  refname={theorem, theorems},
  Refname={Theorem, Theorems}]{theorem}
\declaretheorem[
  name=Lemma,
  refname={lemma, lemmas},
  Refname={Lemma, Lemmas}]{lemma}

\declaretheorem[
  name=Corollary,
  refname={corollary, corollaries},
  Refname={Corollary, Corollaries}]{corollary}
\declaretheorem[
  name=Definition,
  refname={definition, definitions},
  Refname={Definition, Definitions}]{definition}

\newtheorem{proposition}[theorem]{Proposition}
\title{Learning Convex Partitions and Computing Game-theoretic Equilibria from Best Response Queries}

\author{Paul W. Goldberg \thanks{Department of Computer Science, University of Oxford. Emails: \href{mailto:paul.goldberg@cs.ox.ac.uk}{paul.goldberg@cs.ox.ac.uk}, \href{mailto:francisco.marmolejo@cs.ox.ac.uk}{francisco.marmolejo@cs.ox.ac.uk}} \and 
Francisco J. Marmolejo-Coss\'{i}o \footnotemark[1] \thanks{Supported by the Mexican National Council of Science and Technology (CONACyT)}
    }

\date{}

\begin{document}

\maketitle

\begin{abstract}
Suppose that an $m$-simplex is partitioned into $n$ convex regions having disjoint interiors and distinct labels, and we may learn the label of any point by querying it. The learning objective is to know, for any point in the simplex, a label that occurs within some distance $\epsilon$ from that point. We present two algorithms for this task: Constant-Dimension Generalised Binary Search (CD-GBS), which for constant $m$ uses $poly(n, \log \left( \frac{1}{\epsilon} \right))$ queries, and Constant-Region Generalised Binary Search (CR-GBS), which uses CD-GBS as a subroutine and for constant $n$ uses $poly(m, \log \left( \frac{1}{\epsilon} \right))$ queries.

We show via Kakutani's fixed-point theorem that these algorithms provide bounds on the best-response query complexity of computing approximate well-supported equilibria of bimatrix games in which one of the players has a constant number of pure strategies. We also partially extend our results to games with multiple players, establishing further query complexity bounds for computing approximate well-supported equilibria in this setting.

\end{abstract}

\begin{paragraph}{Keywords:} Query protocol, equilibrium computation, revealed preferences
\end{paragraph}

\section{Introduction}\label{sec:introduction}

The computation of game-theoretic equilibria is a topic of long-standing interest in the algorithmic and AI communities. This includes computation in the ``classical'' setting of complete information about a game, as well as settings of partial information, communication-bounded settings, and distributed algorithms (for example, best-response dynamics). A recent line of research has studied computation of equilibria based on query access to players' payoff functions. That work, along with the notion of revealed preferences in economics, inspires the new setting we study here.

We study algorithms that have query access to the players' best-response behaviour: an algorithm may query a mixed-strategy profile (i.e. probability distributions constructed by the algorithm, over each player's pure strategies) and learn the players' best responses. Our focus is on standard bimatrix games, which is arguably the most natural starting-point for an investigation of this new query model. The solution concept of interest is $\epsilon$-approximate Nash equilibria (exact equilibria are typically impossible to find using finitely many such queries, see Corollary \ref{cor:exact-nash-inft}). A basic challenge is to identify algorithms that achieve this goal with good bounds on their query complexity (and also, ideally, their runtime complexity).

In more detail, we assume an $m\times n$ game $G$: a row player has $m$ pure strategies and a column player has $n$ pure strategies. $G$ has two unknown $m\times n$ payoff matrices that represent payoffs to the players for all combinations of pure strategy choices they may make. A query consists of a probability distribution over the pure strategies of one of the players, and elicits an answer consisting of a best response for the other player (i.e. a pure strategy that maximises that player's expected payoff). We seek an $\epsilon$-well-supported Nash equilibrium ($\epsilon$-WSNE): a pair of probability distributions over their pure strategies with the property that any strategy of player $p$ whose expected payoff is more than $\epsilon$ below the value of $p$'s best response, gets probability zero. The general question of interest is: how many queries are needed, as a function of $m,n,\epsilon$.

Using Kakutani's fixed point theorem, we reduce this question to a novel and more geometrical challenge in the design of query protocols. Suppose that the $m$-simplex $\Delta^m$ is partitioned into $n$ convex regions having labels in $[n]=\{1,\ldots,n\}$. When we query a point $x\in\Delta^m$ we are told the label of $x$. How many queries (in terms of $m,n,\epsilon$) are needed in order to ensure that all points in $\Delta^m$ are within $\epsilon$ of a point whose label we know? We show how to achieve this using time and queries polynomial in $\log \epsilon$ and $\max(m,n)$ provided that $\min(m,n)$ is constant. This leads to a polynomial query complexity algorithm for 2-player games, provided that one of the players has a constant number of strategies.

\subsection{Further details}
In essence, we consider partitions of the unit $m$-simplex $\Delta^m$ into $n$ convex polytopes, $P_1,...,P_n$, with disjoint interiors, and aim to approximately learn the partition with access to a membership oracle that for a given $x \in \Delta^m$, returns a polytope to which $x$ belongs. The notion of approximation we study is that of  {\em $\epsilon$-close labellings}, a collection of empirical polytopes, $\{\hat{P}_i\}_{i=1}^n$, such that $\hat{P}_i \subseteq P_i$ for $i=1,...,n$ and $\cup_{i=1}^n \hat{P}_i$ is an $\epsilon$-net of $\Delta^m \subset \mathbb{R}^m$ in the $\ell_2$ norm. 

Note that in one dimension ($m=1$) we can use binary search to solve this problem using $n\log(1/\epsilon)$ queries.
We generalise to higher dimension, exploiting convexity of the regions to reduce query usage in computing $\epsilon$-close labellings.  We present two algorithms for this task: Constant-Dimension Generalised Binary Search (CD-GBS), which for constant $m$ uses $poly(n, \log \left( \frac{1}{\epsilon} \right))$ queries, and Constant-Region Generalised Binary Search (CR-GBS), which uses CD-GBS as a subroutine and for constant $n$ uses $poly(m, \log \left( \frac{1}{\epsilon} \right))$ queries.

This problem derives from the question of how to compute approximate (well-supported) Nash equilibra ($\epsilon$-WSNE) using only best response information, obtained via queries in which the algorithm selects a mixed strategy profile and a player, and receives a best response for that  player to the mixed profile. Via Kakutani's fixed-point theorem \cite{kakutani1941} we reduce this variant of equilibrium computation to finding $\epsilon$-close labellings of polytope partitions. For $m \times n$ games where $m$ is constant (or $n$ equivalently, by symmetry), we show that an $\epsilon$-WSNE can be computed using $poly(n, \log \left( \frac{1}{\epsilon} \right) )$ best response queries.

In addition, we briefly delve into the problem of computing $\epsilon$-WSNE in multiplayer games with best response queries. Unfortunately, as soon as there are more than two players, the geometric connection between computing $\epsilon$-WSNE and learning polytope partitions of the simplex breaks down. Nonetheless, fixed-point techniques from Section \ref{sec:discrete-nash} can still be applied in this setting, and we present a simple algorithm that computes an $\epsilon$-WSNE with a finite query complexity. To be more specific, in a game with $n$ players each having $k$ actions, our algorithm computes an $\epsilon$-WSNE using $O\left( n \left( \frac{nk}{\epsilon} \right)^{nk} \right)$ best response queries.

\subsection{Related Work}
Earlier work in computational learning theory has studied exact learning of geometrical regions over a discretised domain, where algorithms are sought with query complexity logarithmic in a “resolution” parameter and binary search is repeatedly applied in a systematic way \cite{BGGM98}. Goldberg and Kwek~\cite{GK00} specifically study the learnability of polytopes in this context, deriving query efficient algorithms, and precisely classifying polytopes learnable in this setting. These algorithms can be adapted to approximately learn a single polytope with membership queries, but the obtained notion of approximation is not directly applicable to computing $\epsilon$-close labellings. 

The Nash equilibrium (NE) is a fundamental concept in game theory \cite{Nash}. They are guaranteed to exist in finite games, yet computational challenges in finding one abound, most notably, the PPAD-completeness of computing an exact equilibrium even for two-player normal form games \cite{DGP, CDT}. For this reason, query complexity has been extensively used as a tool to differentiate hardness of equilibrium concepts in games. For payoff queries, some notable examples include: exponential lower bounds for randomised computation of exact Nash equilibria and exact correlated equilibria via communication complexity lower bounds in multiplayer games \cite{HM10, HN13}; exponential lower bounds for randomised computation of approximate well-supported equilibria and general approximate equilibria for a small enough approximation factor in multiplayer games \cite{B13}; upper and lower bounds for equilibrium computation in bimatrix games, congestion games \cite{FGGS} and anonymous games \cite{GT14}; upper and lower bounds for randomised algorithms computing approximate correlated equilibria \cite{GR14}. Babichenko et al. have also proved lower bounds in communication complexity for computing $\epsilon$-WSNE for small enough $\epsilon$ in both bimatrix and multiplayer games \cite{BabR}.

Best response queries are a weaker but natural query model which is powerful enough to implement fictitious play, a dynamic first proposed by Brown \cite{Brown}, and proven to converge by Robinson \cite{Robinson} in two-player zero-sum games to an approximate NE. Fictitious play does not always converge for general games where both players have more than two strategies \cite{FL98}. Furthermore, Daskalakis and Pan have proven that the rate of convergence of the dynamic is quite slow in the worst case (with arbitrary tie-breaking) \cite{DFictitious}. Also, beyond non-convergence, the dynamic can have a poor approximation value for general games \cite{GSSV11}. In addition, the relationship between best responses and convex partitions of simplices has been studied by Von Stengel \cite{VS04} in the context of sequential games where one player has to commit to and announce a strategy before playing. 

For a bimatrix game, simple $\epsilon$-close labellings can be constructed by querying best responses at mixed strategies arising as uniform distributions over sufficiently large multisets of pure strategies. As a consequence of our main theorem, best responses to these multiset distributions contain enough information to compute approximate WSNE. This result is in the spirit of \cite{BabBarman} and \cite{LMM03}, who aim to quantify specific $k$ such that some approximate equilibrium arises as a uniform mixture over multisets of size $k$. We note in our scenario that there is also a guaranteed existence of an approximate equilibrium using sufficiently large multisets, however {\em verifying} that a {\em specific} pair of mixed strategies is an approximate WSNE is not straightforward using only best response queries. This is in contrast to the verification of approximate equilibria via utility queries as studied in \cite{BabBarman}.

Separately, we note that the present paper is possibly relevant to the search for a price equilibrium in certain markets. Baldwin and Klemperer study markets consisting of {\em strong-substitutes} demand functions for $N$ different goods available in multiple discrete units \cite{BK16}. These markets are a generalisation of the {\em product-mix auction} of \cite{Kle10}; a basic task is to identify prices at which some desired bundle of the goods is demanded. Consider the space $(\mathbb{R}^+)^N$ of all price vectors. As analysed in \cite{BK16}, a strong-substitutes demand function partitions this price space into convex polytopes, each of which comprises the prices at which some particular bundle of goods is demanded. So, the present paper relates to a setting where price vectors may be queried, and responses consist of demand bundles. The connection is imperfect, since the main objective in the context of \cite{BK16} would be to learn a price at which some target bundle is demanded, rather than the entire demand function. The ideas here may be useful for learning the values that the market has for various bundles.

\section{Preliminaries and Notation}\label{sec:preliminaries}

Our main object of study will be families of polytopes that precisely cover the unit simplex, with the property that any two distinct polytopes from the family are either disjoint, meet at their boundary, or entirely coincide. Throughout, the polytopes we work with are convex.

\begin{definition}[$(m,n)$-Polytope Partition]\label{def:polytope-partition}
A {\em $(m,n)$-polytope partition} consists of a set of $n$ convex polytopes in $\mathbb{R}^m$, $\mathscr{P} = \{P_1,...,P_n\}$, with the following properties:
\begin{itemize}
    \item $\bigcup P_i = \Delta^m = \{x \in \mathbb{R}^m \ | \ \forall i, \ x_i\geq 0, \ \sum_i x_i \leq 1\}$. 
    \item For each $i \neq j$, either $relint(P_i) \cap relint(P_j) = \emptyset$ or $P_i = P_j$, where $relint(H)$ means the relative interior of $H$.
\end{itemize}
\end{definition}

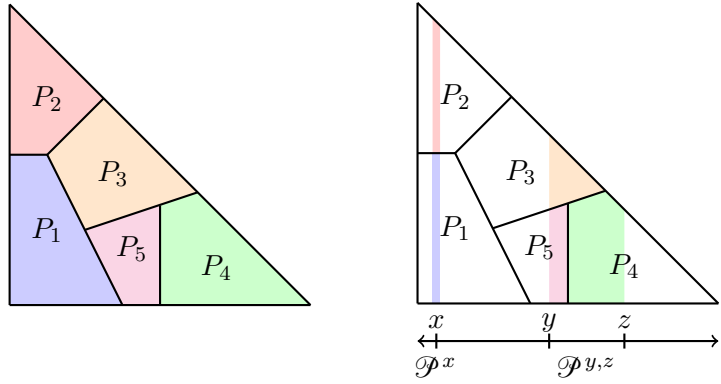
\begin{figure}
\center{
\begin{tikzpicture}[scale=0.5]
\tikzstyle{xxx}=[dashed,thick]

%Colours of Polytopes
\fill[red!20](-8,6)--(-7,6)--(-5.5,7.5)--(-8,10)--cycle; %P1 topmost
\fill[blue!20](-8,6)--(-7,6)--(-5,2)--(-8,2)--cycle; %P2 bottomleft
\fill[orange!20](-6,4)--(-3,5)--(-5.5,7.5)--(-7,6)--cycle; %P3 center one
\fill[magenta!20](-5,2)--(-6,4)--(-4,4.65)--(-4,2)--cycle;
\fill[green!20](-4,2)--(-4,4.65)--(-3,5)--(0,2)--cycle;
\fill[white!20](-8,2)--(0,2)--(0,-0.2)--(-8,-0.2)--cycle;

%border of p1 red region on top
\draw[thick, -](-8,6)--(-7,6)--(-5.5,7.5);

%inner outline
\draw[thick, -](-7,6)--(-5,2); 

%inner outline
\draw[thick, -](-6,4)--(-3,5);

%inner outline
\draw[thick, -](-4,2)--(-4,4.65);

%boundary of overall polytope
\draw[thick, -](-8,2)--(-8,10); %Left side
\draw[thick, -](-8,2)--(0,2); %bottom side
\draw[thick, -](-8,10)--(0,2);

%Labels
\node at (-7,4){$P_1$};
\node at (-7,7.5){$P_2$};
\node at (-5.25,5.5){$P_3$};
\node at (-2.5,3){$P_4$};
\node at (-4.75,3.5){$P_5$};

\end{tikzpicture}
\hspace{1cm}
\begin{tikzpicture}[scale=0.5]
\tikzstyle{xxx}=[dashed,thick]

%Colours of Polytopes
\fill[red!20](-7.6,6)--(-7.4,6)--(-7.4,9.4)--(-7.6,9.6)--cycle; %P1 topmost
\fill[blue!20](-7.6,2)--(-7.4,2)--(-7.4,6)--(-7.6,6)--cycle; %P1 topmost
\fill[orange!20](-4.5,4.5)--(-3,5)--(-4.5,6.5)--cycle; %P3 center one
\fill[magenta!20](-4.5,2)--(-4.5,4.5)--(-4,4.65)--(-4,2)--cycle;
\fill[green!20](-4,2)--(-4,4.65)--(-3,5)--(-2.5,4.5)--(-2.5,2)--cycle;

%border of p1 red region on top
\draw[thick, -](-8,6)--(-7,6)--(-5.5,7.5);

%border of p2 region below 
\draw[thick, -](-7,6)--(-5,2); 

%border of p3 region in center
\draw[thick, -](-6,4)--(-3,5);

%inner outline
\draw[thick, -](-4,2)--(-4,4.65);

%boundary of overall polytope 
\draw[thick, -](-8,2)--(-8,10); %Left side
\draw[thick, -](-8,2)--(0,2); %bottom side
\draw[thick, -](-8,10)--(0,2);

%Labels
\node at (-7,4){$P_1$};
\node at (-7,7.5){$P_2$};
\node at (-5.25,5.5){$P_3$};
\node at (-2.5,3){$P_4$};
\node at (-4.75,3.5){$P_5$};

%axis for labelling cross-section and slice
\draw[thick, <->](-8,1)--(0,1); %bottom side

% location of cross-section
\draw[thick, -](-7.5,0.8)--(-7.5,1.2);

% location of slice
\draw[thick, -](-4.5,0.8)--(-4.5,1.2);
\draw[thick, -](-2.5,0.8)--(-2.5,1.2);

%point labels
\node at (-7.5,1.5){$x$};
\node at (-4.5,1.5){$y$};
\node at (-2.5,1.5){$z$};

%cross-section and slice label
\node at (-7.5,0.3){$\mathscr{P}^x$};
\node at (-3.5,0.3){$\mathscr{P}^{y,z}$};

\end{tikzpicture}

\caption{Polytope partition, cross-section and slices.}\label{fig:poly-part}
}
\end{figure}

\begin{definition}[Cross-sections and Slices]\label{def:cross-section}
Let $P \subset \mathbb{R}^m$ be a polytope and $\pi: \mathbb{R}^m \rightarrow \mathbb{R}$ the projection function into the first coordinate. For $x \in \mathbb{R}$, we define the {\em $x$-cross-section} of $P$ as $P^x = \pi^{-1}(x)\cap P$. For any $I = [x,y] \subset \mathbb{R}$ we define the {\em $[x,y]$-slice} of $P$ as $P^I = P^{x,y} = \pi^{-1}([x,y])\cap P$. Suppose that $\mathscr{P} = \{P_i\}_i$ is an $(m,n)$-polytope partition. The definitions of cross-sections and slices extend to $\mathscr{P}^x = \{P_i^x\}_i$ and $\mathscr{P}^I = \mathscr{P}^{x,y} = \{P_i^{x,y}\}_i$.
\end{definition}

Figure \ref{fig:poly-part} gives a visualisation of these two definitions. Notice that in the same figure, $\mathscr{P}^x$ is essentially a lower-dimensional polytope partition linearly scaled by a factor of $(1-x)$. This however, is not the case in general, as visible in Figure \ref{fig:degenerate-cross}, where $\mathscr{P}^x$ fails the second condition of Definition \ref{def:polytope-partition}. We distinguish between these two scenarios with the following formal definition:

\begin{definition}[Non-Degenerate and Degenerate cross-sections]\label{def:degenerate-cross-section}

Let $\mathscr{P}$ be an $(m,n)$-polytope partition. For $x \in [0,1)$ let $f_x: \mathscr{P}^x \rightarrow \Delta^{m-1}$ be defined as  $f_x(v_1,...,v_m) = \frac{1}{1-x}(v_2,...,v_m)$. If $f_x(\mathscr{P}^x)$ is an  $(m-1,n)$-polytope partition, we say that $\mathscr{P}^x$ is a non-degenerate cross-section. Otherwise we say that $\mathscr{P}^x$ is a degenerate cross-section. 

\end{definition}

The recursive structure of polytope partitions on non-degenerate cross-sections will be crucial to our constructions. Luckily, for any polytope partition, there are only a finite number of points $x \in [0,1)$ that give rise to degenerate cross-sections. Before showing this, we define an important discrete subset of $[0,1]$ given by the projections of vertices of polytopes under $\pi$. 

\begin{definition}[Vertex Critical Coordinates]\label{def:vertex-critical-pts}
For a given polytope $P \subset \mathbb{R}^m$ let $V_P \subset \mathbb{R}^m$ be the vertex set of $P$. Define the set of {\em vertex critical coordinates} as $C_P = \pi(V_P) \subset \mathbb{R}$. If $\mathscr{P} = \{P_i\}_{i=1}^n$ is an $(m,n)$-polytope partition, then we extend our definition to define $C_{\mathscr{P}} = \bigcup_{i=1}^n C_{P_i} \subset[0,1]$ as the vertex critical coordinates of $\mathscr{P}$. 

\end{definition}

\begin{lemma}\label{lemma:degenerate-cross-section}
Suppose that $\mathscr{P}$ is an $(m,n)$-polytope partition and that $x \in [0,1) \setminus C_{\mathscr{P}}$. Then $\mathscr{P}^x$ is non-degenerate.
\end{lemma}

\begin{proof}
First we show that if $P \subset \mathbb{R}^m$ is an arbitrary polytope and $x \in \mathbb{R} \setminus C_P $ then $relint(P^x) = relint(P)\cap \pi^{-1}(x)$.

First of all, we notice that $P \neq P^x$ since we have assumed that $x$ is not the projection of a vertex of $P$. Suppose that the affine dimension of $P$ is $k \leq m$ so that $P$ is full dimensional in the affine subspace $H$ of dimension $k$. Let $z \in relint(P^x) \subset P^x$. Clearly $\pi(z) = x$, hence we simply need to show that $z \in relint(P)$. Suppose that this is not the case, then $z$ lies on some boundary hyperplane to $P$ in $H$. Call this boundary hyperplane $D$. $D$ cannot lie entirely in $\pi^{-1}(x)$ due to the fact that $x$ is not a critical coordinate. It follows that $D \cap \pi^{-1}(x)$ is thus a boundary hyperplane to $P^x$. This contradicts the fact that $z \in relint(P^x)$, thus establishing the fact that $z \in relint(P) \cap \pi^{-1}(x)$.

Suppose that $z \in relint(P) \cap \pi^{-1}(x)$. Since $relint(P) \subset P$, we know that $z \in P_x$. Furthermore, $z \in relint(P)$ means that for some $\epsilon > 0$, the $k$-dimensional ball $B_\epsilon(z) \cap H$ is entirely contained in $P$. Clearly this also holds for the $k-1$ dimensional ball $B_\epsilon(z) \cap H \cap \pi^{-1}(x)$, establishing the fact that $z \in relint(P^x)$. 

Let us return to the lemma statement which involves a polytope partition $\mathscr{P}$ instead of a single polytope $P$. If $x \in [0,1) \setminus C_{\mathscr{P}}$ then we have shown $relint(P_i^x) = relint(P_i)\cap \pi^{-1}(x)$ for all $P_i \in \mathscr{P}$, which from the fact that $\mathscr{P}$ satisfies the second condition of Definition \ref{def:polytope-partition} establishes the fact that $\mathscr{P}^x$ also satisfies this second condition. The fact that $\mathscr{P}^x$ satisfies the first condition of Definition \ref{def:polytope-partition} trivially follows from the fact that $\mathscr{P}$ covers $\Delta^m$ as per the first condition of Definition \ref{def:polytope-partition}. 

\end{proof}

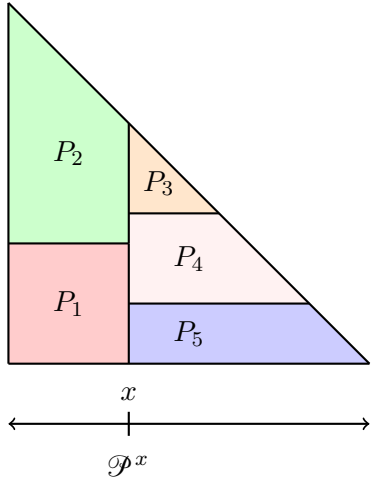
\begin{figure}
\center{
\begin{tikzpicture}[scale=0.8]
\tikzstyle{xxx}=[dashed,thick]

%Partition Colors
\fill[red!20](2,2)--(2,4)--(4,4)--(4,2)--cycle; %red right
\fill[green!20](2,4)--(4,4)--(4,6)--(2,8)--cycle; %green right 
\fill[blue!20](4,2)--(8,2)--(7,3)--(4,3)--cycle; %blue right
\fill[pink!20](4,3)--(7,3)--(5.5,4.5)--(4,4.5)--cycle;
\fill[orange!20](4,4.5)--(5.5,4.5)--(4,6)--cycle;

%x-axes
\draw[thick, -](2,2)--(8,2);

%Left y-axes
\draw[thick, -](2,2)--(2,8);

%Triangle outlines
\draw[thick, -](2,8)--(8,2);

%polytope boundaries
\draw[thick, -](4,4)--(4,6);
\draw[thick, -](4,4)--(4,2);
\draw[thick, -](2,4)--(4,4);
\draw[thick, -](4,4.5)--(5.5,4.5);
\draw[thick, -](4,3)--(7,3);

%Labels
\node at (3,3){$P_1$};
\node at (3,5.5){$P_2$};
\node at (4.5,5){$P_3$};
\node at (5,3.75){$P_4$};
\node at (5,2.5){$P_5$};

%axis for labelling cross-section and slice
\draw[thick, <->](2,1)--(8,1); %bottom side

% location of cross-section
\draw[thick, -](4,0.8)--(4,1.2);

%point labels
\node at (4,1.5){$x$};

%cross-section and slice label
\node at (4,0.3){$\mathscr{P}^x$};

\end{tikzpicture}
\caption{Degenerate cross-section at $x$}\label{fig:degenerate-cross}
}
\end{figure}

\subsection{Query Oracle Models}

We study two natural query oracle models for polytope membership in any $\mathscr{P}$.

\begin{definition}[Membership Query Oracles for Polytope Partitions]\label{def:query-oracle}
Any $(m,n)$-polytope partition, $\mathscr{P}$ has the following membership query oracles:
\begin{itemize}
    \item Lexicographic query oracle: $Q_\ell: \Delta^m \rightarrow [n]$, which for a given $y$ returns the smallest index of polytope to which $y$ belongs, namely $Q_\ell(y) = \min \{i \in [n] \ | \ y \in P_i\}$.
    \item Adversarial query oracle(s): $Q_A: \Delta^m \rightarrow [n]$, which can return any polytope to which $y$ belongs. Namely $Q_A$ is any function such that $Q_A(y) \in \{i \in [n] \ | \ y \in P_i\}$ for all $y \in \Delta^m$.     
\end{itemize}
\end{definition}
When we wish to refer to an arbitrary oracle from the above models, we use the notation $Q$. Before continuing, we also clarify that for $A,B \subseteq \mathbb{R}^n$, we denote $Conv(A,B) \subseteq \mathbb{R}^m$ as the convex combination of the two sets. In addition, if $A_i \subseteq \mathbb{R}^m$ is an indexed family of sets with $i = 1,...,r$, we denote $Conv(A_i \ | i =1,...,r) \subseteq \mathbb{R}^n$ as the convex combination of all $A_i$.

\subsection{\texorpdfstring{$\epsilon$}{}-close Labellings}
Upon making queries to $Q$, we can infer labels of $x \in \Delta^m$ by taking convex combinations. We abstract this notion in the following definition. 

\begin{definition}[Empirical Polytopes and Labellings]\label{def:implicit-labelling}
Suppose that $\mathscr{P}$ is an $(m,n)$-polytope partition
and $S \subset \Delta^m$ is a finite set for which queries to $Q$ have been made. Let $\hat{P}_i = Conv(\{x \in S \ | \ Q(x) = i\}) \subset P_i$. We say each $\hat{P}_i$ is an empirical polytope of $P_i$ and that $\hat{\mathscr{P}} = \{\hat{P}_i\}$ is an empirical labelling of $\mathscr{P}$. Furthermore, we use the notation $\hat{P}_\bot = \Delta^m \setminus \cup_{i=1}^n \hat{P}_i$. to refer to points in $\Delta^m$ unlabelled under $\hat{\mathscr{P}}$. 
\end{definition}

An $\epsilon$-net in the $\ell_2$ norm for $ \Delta^m \subset \mathbb{R}^m$ is a set $N^m_\epsilon \subseteq \Delta^m$ with the property that for all $x \in \Delta^m$, there exists a $y \in N^m_\epsilon$ such that $\|x - y\|_2 \leq \epsilon$. Our learning goal is to use query access to an oracle, $Q$, to compute an empirical labelling $\hat{\mathscr{P}}$ such that $\cup_{i=1}^n \hat{P}_i$ is an $\epsilon$-net of $\Delta^m$. 

\begin{definition}[$\epsilon$-close Labelling]\label{def:eps-close-label-thin}
Suppose that $\epsilon \geq 0$ and that $\hat{\mathscr{P}}$ is an empirical labelling for $\mathscr{P}$. If $\cup_{i=1}^n \hat{P}_i$ is an $\epsilon$-net of $\Delta^m \subset \mathbb{R}^m$ in the $\ell_2$ norm, we say that $\hat{\mathscr{P}}$ is an {\em $\epsilon$-close labelling} of $\mathscr{P}$.
\end{definition}

Although $\epsilon$-close labellings are defined for polytope partitions, we extend our terminology to also encompass slices of polytope partitions. As such, when we mention computing an $\epsilon$-close labelling of $\mathscr{P}^{x,y}$, we mean an empirical labelling of $\mathscr{P}^{x,y}$ (in the same vein as Definition \ref{def:implicit-labelling}) with the property that the union of its empirical polytopes forms an $\epsilon$-net of $(\Delta^m)^{x,y}$.

\subsection{Learning in Thickness to Learning in Distance}

\begin{definition}[Thickness of Sets]\label{def:thickness}
Suppose that $Z \subseteq \mathbb{R}^m$ is a set. We define the {\em thickness} of $Z$ as the radius of the largest $\ell_2$ ball fully contained in $Z$ and we denote it by $\tau(Z) = \sup \{\delta \geq 0 \ | \ \exists x \in Z \text{ with } B_\delta(x) \subseteq Z\}$ where $B_\delta(x) = \{ y \in \mathbb{R}^m \ | \ \|x - y\|_2 \leq \delta\}$. In the language of convex geometry, $\tau(Z)$ is the depth of the Chebyshev centre of $Z$.
\end{definition}

For a polytope partition $\mathscr{P}$, if $\hat{\mathscr{P}}$ is an $\epsilon$-close labelling, then $\tau(\hat{P}_\bot)  \leq \epsilon$, but the converse does not hold in general.
%as can be seen in figure \ref{fig:example-thick-thin}. 
Even though $\hat{P}_\bot$ may be of small thickness, if it contains vertices of $\Delta^m$, these vertices may be far from labelled points. The following results lead up to Lemma \ref{lemma:thickness-to-distance}, a slightly weaker version of the converse. Lemma \ref{lemma:thickness-to-distance} shows that if we are able to learn an empirical labelling where the set of unlabelled points is of small enough thickness, then we will in fact have succeeded in learning an $\epsilon$-close labelling, where any unlabelled point is close in distance to a labelled point.

\begin{lemma}\label{lemma:cones-nets}
Let $P \subset \mathbb{R}^m$ be a full-dimensional polytope with $Diam(P) = \sup_{x,y \in P} \|x - y\|_2$.
\begin{itemize}
    \item If $A \subsetneq P$ and $\gamma > \left(\frac{Diam(P)}{\tau(P)} \right) \tau(A)$, then $B_\gamma(x) \cap \left( P \setminus A \right) \neq \emptyset$ for all $x \in A$.
    \item If  $A \subseteq P$ is such that $int(P) \setminus A \neq \emptyset$ ($int(P)$ refers to the interior of $P$) and $\gamma > \left(\frac{Diam(P)}{\tau(P)} \right) \tau(A)$, then $B_\gamma(x) \cap \left( int(P) \setminus A \right) \neq \emptyset$ for all $x \in A$. 
\end{itemize}
\end{lemma}

\begin{proof}
The proof of the first claim follows by considering the picture given in Figure \ref{fig:thickness-to-distance}. Pick an arbitrary $x \in A$. Due to the definition of thickness, there exists some $v \in P$ such that $B_{\tau(P)}(v) \subset P$. Consider the convex combination, $Conv(x, B_{\tau(P)}(v)) \subset P$. The furthest point in this convex combination from $x$ is at the other extreme of $B_{\tau(P)}(v)$ from $x$, and we denote the distance between these two points by $z = \sup_{a \in B_{\tau(P)}(v)} \|x - a\|_2 \leq Diam(P)$.
By similarity however, it now follows that if we consider $B_\gamma(x)$, and the fact that $Diam(P) \geq z$, a similar inscribed sphere of radius strictly greater than $\tau(A)$ will exist within $F = B_{\gamma}(x) \cap Conv(x, B_{\tau(P)}(v)) \subset B_{\gamma}(x) \cap P$. By definition, $\tau(F) > \tau(A)$. It follows that $F \not\subset A$, which proves the claim as $F \subset B_{\gamma}(x)$.  

As for a proof of the second claim, it follows by considering the same picture above and noticing that $int(Conv(x, B_{\tau(P)}(v))) \subset int(P)$ as well as the fact that the former set is non-empty since $P$ is of full dimension. 
\end{proof}

\begin{figure}[h]
\center{
\begin{tikzpicture}[scale=0.8]
\tikzstyle{xxx}=[dashed,thick]

%shading before everything else
\fill[blue!20](-3,-0)--(4.5,2.5)--(4.5,-2.5)--cycle; %triangular part
\draw[red,fill=blue!20] (5,0) circle (2.5cm);

%central point
\filldraw (-3,0) circle (2pt);

%central point for guaranteed ball
%\filldraw (5,0) circle (2pt);

%guaranteed ball
\draw (5,0) circle (2.5cm);

%cone outline
\draw[thick, -](-3,0)--(4.5,2.5);
\draw[thick, -](-3,0)--(4.5,-2.5);

%ball around x
\draw (-3,0) circle (3cm);

%insribed circle
\draw (-0.69,0) circle (0.7cm);

%axes for scale
\draw[thick, <->](-3,-3.5)--(7.5,-3.5);
\draw[thick, <->](-3,3.5)--(0,3.5);

%radii for scale
\draw[thick, <->] (5,0)--(5,2.5);
\draw[thick, <->] (-0.69,0)--(-0.69,0.7);

%labels of distances
\node at (-1.4,0.95) {$\tau(A)$};
\node at (5.6,1.2) {$\tau(P)$};
\node at (2,-4) {$z$};
\node at (-1.5,4) {$\gamma = \frac{z\tau(A)}{\tau(P)}$};

%label of membership in P
\node at (8.5,0) {$\subset P$};

%node for x
\node at (-3.5,0) {$x$};
%node for v
\node at (5,-0.25) {$v$};

\end{tikzpicture}
\caption{Proof of Lemma \ref{lemma:cones-nets}
\label{fig:thickness-to-distance}
}
}
\end{figure}
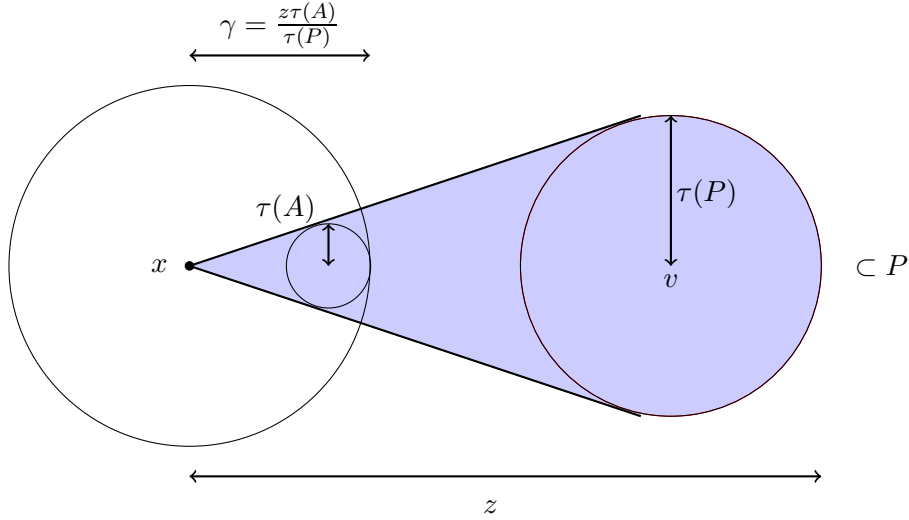

\begin{lemma}\label{lemma:thick-diam-simplex}
$Diam(\Delta^m) = \sqrt{2}$ and $\tau(\Delta^m) \geq \frac{1}{m + \sqrt{m}}$. 
\end{lemma}

\begin{proof}
 For the first part of the statement, let us fix an $x \in \Delta^m$. If we consider the function $f_x(z) = \|x - z\|_2^2$, then this function is differentiable and clearly achieves local maximum when $z$ is a vertex of $\Delta^m$. It thus follows that the distance between two points in $\Delta^m$ is maximised when both are vertices. This in turn is maximal when both points are vertices not equal to the zero vector, in which case they are at distance $\sqrt{2}$ from each other.

As for the second part of the claim, we explicitly construct an inscribed sphere of the desired radius. Let $\lambda = \frac{1}{m + \sqrt{m}}$, and define $x = \lambda \vec{1} \subset \Delta^m$. Clearly $B_\lambda(x)$ is tangent to $\Delta^m$ on axis-aligned faces (defined by the set of all $z$ such that  $\pi_i(z) = 0$ in the positive orthant). The remaining face is given by the set of $z$ in the positive orthant such that $\|z\|_1 = 1$. The most extremal point of $B_\lambda(x)$ in the direction of this face is given by $\frac{1}{m}\vec{1}$, hence the sphere is inscribed in $\Delta^m$.
\end{proof}

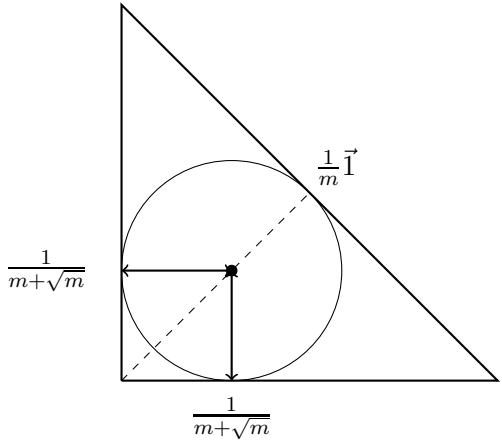
\begin{figure}[h]
\center{
\begin{tikzpicture}[scale=1]
\tikzstyle{xxx}=[dashed,thick]

%outline and dashed lines
\draw[thick, -] (0,0)--(0,5)--(5,0)--(0,0);
\draw[dashed] (0,0)--(2.5,2.5);

%inscribed circle
\filldraw (1.464,1.464) circle (2pt);
\draw (1.464,1.464) circle (1.464cm);

%labels for inscribed triangle
\draw[thick, <->] (1.464,1.464)--(1.464,0);
\draw[thick, <->] (0,1.464)--(1.464,1.464);

%node labels
\node at (-1,1.464) {$\frac{1}{m + \sqrt{m}}$};
\node at (1.464,-0.5) {$\frac{1}{m + \sqrt{m}}$};
\node at (2.85,2.85) {$\frac{1}{m}\vec{1}$};

\end{tikzpicture}
\caption{Proof of Lemma \ref{lemma:thick-diam-simplex}
\label{fig:simplex-thickness}
}
}
\end{figure}

\begin{lemma}\label{lemma:thickness-to-distance}
Suppose that $\mathscr{P}$ is an $(m,n)$-polytope partition. Furthermore suppose that $\hat{\mathscr{P}}$ is an empirical labelling with $\tau(\hat{P}_\bot) < \epsilon$. For any $\gamma > \sqrt{2}(m + \sqrt{m})\epsilon$, it follows that $\hat{\mathscr{P}}$ is a $\gamma$-close labelling. In particular, if $\gamma > 4m\epsilon$, the claim also holds.\footnote{An identical result which may be of separate interest holds if we consider partitions of arbitrary $m$-dimensional convex polytopes (not just $\Delta^m$ as per the definition of polytope partitions). As long as we can bound the thickness and diameter of the ambient convex polytope, learning in thickness translates to learning in distance.}
\end{lemma}

\begin{proof}
From Lemma \ref{lemma:thick-diam-simplex}, we know that $\tau(\Delta^m) \geq \frac{1}{m + \sqrt{m}}$ and $Diam(\Delta^m) = \sqrt{2}$. Suppose that $x \in \hat{P}_\bot$. From Lemma \ref{lemma:cones-nets} our choice of $\gamma$ implies $B_\gamma(x) \cap (\Delta_m \setminus \hat{P}_\bot) \neq \emptyset$. This in turn means that $\hat{\mathscr{P}}$ is a $\gamma$-close labelling. As for the final claim, this holds since $m \geq 1$.
\end{proof}

\section{Constant-Dimension Generalised Binary Search for \texorpdfstring{$Q_\ell$}{}} 

Let us first build some intuition for why generalisations of binary search lead to query efficient algorithms for computing $\epsilon$-close labellings of $(m,n)$-polytope partitions. 

Finding an $\epsilon$-close labelling of a $(1,n)$-polytope partition using a lexicographic oracle is the same as approximately learning $n$ sub-intervals of $[0,1]$. Using binary search techniques and an optimal $O(n \log (\frac{1}{\epsilon}))$ queries, we can compute an $\epsilon$-close labelling.
%Furthermore, this amount of queries is also optimal, as decision tree approaches show. 

Query efficiency comes from the fact that if $x,y$ have the same label, it becomes unnecessary to further query any point in $[x,y]$. To be more specific, unless $[x,y]$ contains the boundary of a sub-interval, all labels can be inferred within $[x,y]$. Boundary points of intervals thus serve as ``critical points'' with respect to the query oracle $Q_\ell$, where the information it provides changes.

We will use a higher-dimensional analogue of this property at the core of CD-GBS. At a high level, suppose that we have an $(m,n)$-polytope partition that we want to learn via queries and an algorithm for computing arbitrary $\epsilon$-close labellings of $(m-1,n)$-polytope partitions. We can use this algorithm as a subroutine on the cross-sections of two coordinates $x \neq y$ and ask whether the convex combination of these two $\epsilon$-close labellings will itself result in a $g(\epsilon)$-close labelling of $\mathscr{P}^{x,y}$ (recall Definition \ref{def:cross-section}) for a reasonable $g$. 

Suppose that we could compute $0$-close labellings (i.e. perfectly recover a polytope partition), it is clear that if we let $\mathscr{P}_V$ be the set of all vertices of all $P_i$ in the polytope partition, then $\pi(\mathscr{P}_V)$ is a suitable set of critical points (not necessarily the smallest one though) in the sense that if $[x,y] \cap \pi(\mathscr{P}_V) = \emptyset$, then the convex combination of both lower-dimensional $0$-close labellings for $\mathscr{P}^x$ and $\mathscr{P}^y$ will result in a $0$-close labelling for $\mathscr{P}^{x,y}$. Taking the contrapositive of this, if the convex combination does not result in a $0$-close labelling ---a condition which can be verified--- then we know there is a critical point in $[x,y]$. Thus we recover a binary search mechanism, whereby we can isolate critical points up to a desired tolerance $\epsilon$.

\subsection{Warm-up: Learning Slices of Single Polytopes}
We set up important groundwork by focusing on arbitrary polytopes $P \subset \mathbb{R}^m$. We let $\pi:\mathbb{R}^m \rightarrow \mathbb{R}$ be the projection function introduced in Definition \ref{def:cross-section}, and we recall Definition \ref{def:vertex-critical-pts} regarding the vertex critical coordinates of $P$ denoted by $C_P$.

\begin{lemma} \label{lemma:perfect-fleshing}
Suppose that $x,y \in \mathbb{R}$ are such that $[x,y] \bigcap C_P = \emptyset$. Then taking convex hulls of cross-sections we get $Conv(P^x, P^y) = P^{x,y}$.
\end{lemma}

\begin{proof}
 $[x,y] \cap C_{P} = \emptyset$ implies the vertices of the polytope $P^{x,y}$ lie in $\mathscr{P}_x$ and $\mathscr{P}_y$. Since the convex hull of the set of all vertices of a bounded polytope is the polytope itself, the claim follows.
\end{proof}

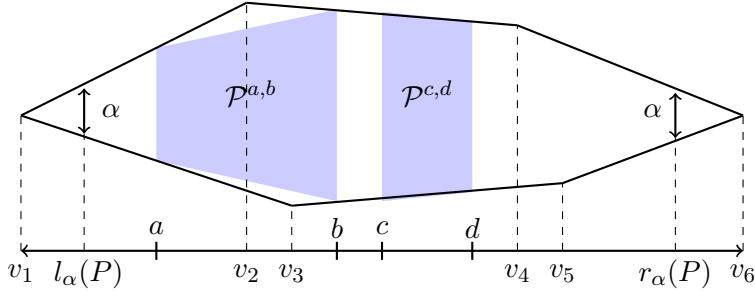
\begin{figure}[h]
\center{
\begin{tikzpicture}[scale=0.6]
\tikzstyle{xxx}=[dashed,thick]

%color of fleshed slice
\fill[blue!20](-4,2.1)--(-2,2.3)--(-2,6.1)--(-4,6.3)--cycle; %P4 bottom top right

%color of incomplete slice
\fill[blue!20](-9,3)--(-5,2.1)--(-5,6.35)--(-9,5.5)--cycle; %P4 bottom top right

%boundary of overall polytope
\draw[thick, -](-6,2)--(0,2.5); %bottom side
\draw[thick, -](0,2.5)--(4,4); % right bottom
\draw[thick, -](4,4)--(-1,6); % right top
\draw[thick, -](-1,6)--(-7,6.5); % top side
\draw[thick, -](-7,6.5)--(-12,4); %left top
\draw[thick, -](-12,4)--(-6,2); %left bottom

%dashed lines for critical points
\draw[dashed](-12,1)--(-12,4);
\draw[dashed](-7,1)--(-7,6.5);
\draw[dashed](-6,1)--(-6,2);
\draw[dashed](-1,1)--(-1,6);
\draw[dashed](0,1)--(0,2.5);
\draw[dashed](4,1)--(4,4);

%dashed line for thickness critical points
\draw[dashed](-10.6,1)--(-10.6,3.6);
\draw[dashed](2.5,1)--(2.5,3.5);

%axis for labelling cross-section and slice
\draw[thick, <->](-12,1)--(4,1); %bottom side

%point labels
\node at (-12,0.5){$v_1$};
\node at (-7,0.5){$v_2$};
\node at (-6,0.5){$v_3$};
\node at (-1,0.5){$v_4$};
\node at (0,0.5){$v_5$};
\node at (4,0.5){$v_6$};

%thickness coordinate labels
\node at (-10.5,0.5){$l_\alpha(P)$};
\node at (2.5,0.5){$r_\alpha(P)$};

% location of slice
\draw[thick, -](-4,0.8)--(-4,1.2);
\draw[thick, -](-2,0.8)--(-2,1.2);

%point labels
%\node at (-5.275,1.5){$x$};
\node at (-4,1.5){$c$};
\node at (-2,1.5){$d$};

%cross-section and slice label
%\node at (-5.275,0.3){$\mathcal{P}_x$};
\node at (-3,4.5){$\mathcal{P}^{c,d}$};

%FOR THE INCOMPLETE SLICE
% location of slice
\draw[thick, -](-9,0.8)--(-9,1.2);
\draw[thick, -](-5,0.8)--(-5,1.2);

%point labels
%\node at (-5.275,1.5){$x$};
\node at (-9,1.5){$a$};
\node at (-5,1.5){$b$};

%cross-section and slice label
\node at (-6.9,4.5){$\mathcal{P}^{a,b}$};

%axis for left thickness coordinate
\draw[thick, <->](-10.6,3.6)--(-10.6,4.6); %bottom side
\node at (-10,4.1){$\alpha$};

%axis for right thickness coordinate
\draw[thick, <->](2.5,3.5)--(2.5,4.5); %bottom side
\node at (2,4.1){$\alpha$};

\end{tikzpicture}
\caption{$Conv(P^a, P^b) \neq P^{a,b}$ and $Conv(P^c,P^d) = P^{c,d}$}
\label{fig:perfect-fleshing-2}
}
\end{figure}

This property of polytopes whereby convex combinations give rise to complete information except when traversing a discrete set of critical points (visualised in Figure %\ref{fig:perfect-fleshing}, 
\ref{fig:perfect-fleshing-2}) is critical to CD-GBS. With query access to polytopes however, we no longer fully recover $P_x$ perfectly, but instead an approximation given by an $\epsilon$-close labelling, $\hat{P}_x$. It becomes more subtle to show that by taking convex hulls of $\hat{P}_x$ and $\hat{P}_y$, we recover the desired information along $[x,y]$.

\subsection{Necessary Machinery}\label{sec:machinery}
We delve into the specifics of CD-GBS by defining some important machinery. We recall our notion of thickness in Definition \ref{def:thickness}, and see that it satisfies a sub-additivity property when the sets being considered are convex polytopes:

\begin{lemma}\label{lemma:subadd}
Let $P_1,..,P_k \subseteq \mathbb{R}^m$ be convex polytopes. Then $\tau \left( \cup_i P_i \right) < \frac{10}{3}(\sum_i \tau(P_i)) (m+1)^{3/2}$.
\end{lemma}

\begin{proof}
Let $R = \frac{10}{3}(\sum_i \tau(P_i)) (m+1)^{3/2}$. Suppose that $x \in \cup_i P_i$. We will show that $B_R(x)$ cannot be a subset of $\cup_i P_i$ via a volume argument. For this proof, we will let $V(A)$ denote the volume of the set $A \subset \mathbb{R}^m$. We will also let $S(m,R)$ denote the volume of the hypersphere in $m$ dimensions of radius $R$.

First of all, we need to show that for a given $P_i$, we have the following volume bound:
$$
V(P_i \cap B_R(x)) \leq 2\tau(P_i) S(m-1,R).
$$
This follows from Fritz John's Theorem \cite{Fritz} especially as referenced in  \cite{BallConvex}.
The statement of this theorem says that if $K\subseteq \mathbb{R}^m$ is a convex body, then there exists a unique ellipsoid of maximal volume $\mathcal{E}\subseteq K$, with the property that $\mathcal{E} \subseteq P \subseteq m \mathcal{E}$. Any higher-dimensional ellipsoid has at most $m$ axes of symmetry, and for $\mathcal{E}$, it must be the case that the smallest axis is at most the thickness of the convex body: $\tau(K)$ (Otherwise there would be a sphere of radius larger than $\tau(K)$ inscribed in $\mathcal{E}$, contradicting the definition of thickness). Furthermore, since $K \subseteq m \mathcal{E}$, the projection of $K$ onto this axis must be contained in a segment of length at most $2m \tau(K)$. This means that if we take an arbitrary polytope $P_i$, there exist two parallel supporting hyperplanes to $P_i$, call them $H_1$ and $H_2$ that are at most $2m\tau(P_i)$ apart. Call the convex body between these halfspaces $H$. Since the majority of the mass of a hypersphere is contained around its centre, it follows that the volume of the intersection of $H$ with $B_R(x)$ is maximised when $x$ is equidistant from $H_1$ and $H_2$. Furthermore, the volume of this cross-section is bounded by the distance between $H_1$ and $H_2$ multiplied by $S(m-1,R)$ which is at most $2\tau(P_i) S(m-1,R)$. Since $V(P_i \cap B_R(x)) \leq V(H \cap B_R(x))$, the claim holds.

Now if we take a union bound over all $i$, we get $V((\cup_i P_i) \cap B_R(x)) \leq 2m\sum_i\tau(P_i) S(m-1,R)$. If it were the case that the right hand side were strictly less than $S(m,R)$, we would have found an $R$ such that $B_R(x)$ contains points not contained in any $P_i$. To this end, we use the following ratio:
$$
\frac{S(m,R)}{S(m-1,R)} = R \sqrt{\pi} \frac{\Gamma(\frac{m+1}{2})}{\Gamma(\frac{m+2}{2})} \geq 0.6 R (m+1)^{-1/2}.
$$
The inequality uses Stirling's formula for the gamma function. We can therefore see that with our value $R > \frac{10}{3}(\sum_i \tau(P_i)) (m+1)^{3/2}$, we get the desired volume bound.
\end{proof}

For a given polytope partition $\mathscr{P} = \{P_i\}_i$, it will be important to establish thickness bounds on $P_i$ at specific cross-sections. 

\begin{definition}[$\alpha$-Critical Coordinates]\label{def:critical-thick}
Let $P \subset \mathbb{R}^m$ be a polytope. For $\alpha > 0$, we define $l_\alpha(P) = \inf \{ x\in \mathbb{R} \ | \ \tau(P^x) \geq \alpha\}$ and $r_\alpha(P) = \sup \{ x\in \mathbb{R} \ | \ \tau(P^x) \geq \alpha\}$ so that $\forall z \in \mathbb{R}$, $\tau(P^z) \geq \alpha$ if and only if $z \in [l_\alpha(P), r_\alpha(P)]$ (Here thickness is with respect to the natural embedding of $P^x$ in $\mathbb{R}^{m-1}$). These are called {\em $\alpha$-critical coordinates} for $P$. 
\end{definition}

The previous definition allows us to associate to each polytope $P_i$ a segment of $[0,1]$ within which cross-sections of $P_i$ are thick above a threshold. By combining this with Definition \ref{def:vertex-critical-pts} we get the correct notion of critical coordinates mentioned at the beginning of Section \ref{sec:discrete-nash}. 

\begin{definition}[Critical Coordinates of a $(m,n)$-Polytope Partition]\label{def:critical-partition}
Suppose that $\mathscr{P} = \{P_1,...,P_n\}$ is an $(m,n)$-polytope partition. For $\alpha > 0$, we let $C_\mathscr{P}^\alpha$ be the union of the sets of all vertex critical coordinates of all $P_i$ as defined in Definition \ref{def:vertex-critical-pts}, and the set of all $\alpha$-critical coordinates for all $P_i$ as in Definition \ref{def:critical-thick}. Specifically, $C_\mathscr{P}^\alpha = \left( \cup_i C_{P_i} \right) \bigcup \left( \cup_i \{l_\alpha(P_i), r_\alpha(P_i) \}   \right)$.
\end{definition}

As mentioned in the beginning of this section, CD-GBS clusters queries around critical coordinates (up to a desired tolerance). For this reason it is important to bound the number of critical coordinates in a given $(m,n)$-Polytope partition.

\begin{lemma}\label{lemma:bound-critical-cardinality}
If $\mathscr{P}$ is a $(m,n)$-polytope partition $|C_\mathscr{P}^\alpha| \leq \binom{n + m} {m} + 2n$.
\end{lemma}

\begin{proof}
For any given $(m,n)$-polytope partition, $\mathscr{P}$, if a vertex occurs, it must be the case that out of the $n$ polytopes in $\mathscr{P}$ and $m$ boundary halfspaces of $\Delta^m$, $m$ of them meet. Furthermore, each collection of $m$ polytopes and boundary halfspaces can give rise to only one vertex (which can be seen as a consequence of the fact that vertices are points in $\Delta^m$). It follows that the set of all vertex critical coordinates is at most $\binom {n + m} {m}$ and the first part of the bound holds. As for the second half, there are at most two $\alpha$-critical coordinates per $P_i$, which completes the expression above.
\end{proof}

With this machinery in hand, we are in a position to prove the main result necessary to demonstrate correctness of CD-GBS. We show that if $x,y \in [0,1]$ are such that $[x,y]$ contains no critical coordinates, then computing sufficiently fine empirical labellings of $\mathscr{P}^x$ and $\mathscr{P}^y$ with $Q_\ell$ will contain enough information to compute an $\epsilon$-close labelling of $\mathscr{P}^{x,y}$ by simply taking convex combinations of the empirical labellings at both cross-sections.

\begin{lemma}\label{lemma:crux-for-gbs}
Given $m,n,\epsilon>0$ let $\alpha = \frac{\epsilon}{20 n m^{5/2}} $ and $\beta =  \frac{\epsilon^2}{85 n m^{5/2}} $.
Suppose that $\mathscr{P}$ is an $(m,n)$-polytope partition and that the following hold:
\begin{itemize}
    \item $x,y \in [0,1]$ are such that $x < y \leq 1 - \frac{\epsilon}{3}$.
    \item $[x,y] \cap C_{\mathscr{P}}^\alpha = \emptyset$.
    \item $\hat{\mathscr{P}}^x$ and $\hat{\mathscr{P}}^y$ are empirical labellings of $\mathscr{P}^x$ and $\mathscr{P}^y$ computed via $Q_\ell$, such that $\cup_i \hat{P}_i^x$ and $\cup_j \hat{P}_j^y$ are $\beta$-nets for $(\Delta^m)^x$ and $(\Delta^m)^y$ respectively.
\end{itemize}
Then $\bigcup_i Conv(\hat{P}^x_i, \hat{P}^y_i)$ is an $\epsilon$-net of $(\Delta^m)^{x,y}$.
%Then for any $w \in (\Delta^m)^{x,y}$, there exists a $w' \in Conv(\hat{\mathscr{P}}^x, \hat{\mathscr{P}}^y)$ such that $|w - w'| < \epsilon$.
\end{lemma}

\begin{proof}
Let us define the following:
$$
U = \{i \in [n] \ | \ [l_\alpha(P_i), r_\alpha(P_i)] \cap [x,y] = \emptyset \}
$$
$$
V = \{i \in [n] \ | \ [x,y] \subsetneq [l_\alpha(P_i), r_\alpha(P_i)]\}
$$
We call $U$ the set of $\alpha$-insignificant polytopes and $V$ the set of $\alpha$-significant polytopes. From the fact that $[x,y]$ contains no critical coordinates, we know that $U \cup V = [n]$ and from Lemma \ref{lemma:degenerate-cross-section}, we also know that for all $z \in [x,y]$, $\mathscr{P}^z$ is non-degenerate. We proceed by proving the following claims:
\begin{enumerate}
    \item $V \neq \emptyset$. 
    \item Any point in the cross-section of an $\alpha$-insignificant polytope is $\frac{2\epsilon}{3}$ close to an $\alpha$-significant polytope (within that same cross-section).
    \item If $e \in P_j^x \setminus \left( \bigcup_{i=1}^n \hat{P}_i^x \right)$, and $j \in V$, then there exists a $e' \in \hat{P}_j^x$ such that $\|e-e'\|_2 < \frac{\epsilon}{3}$. 
    \item If $w \in P_j^z$ for some $j \in V$ and $z \in [x,y]$, then there exists a $w' \in Conv(\hat{P}_j^x, \hat{P}_j^y) \cap \pi^{-1}(z)$ such that $\|w-w'\|_2 < \frac{\epsilon}{3}$.
\end{enumerate} 
(2) and (4) suffice to prove the theorem. To see this, suppose that $w \in (\Delta^m)^{x,y}$. This means that $w \in P_i^z$ for some $i \in [n]$ and $z \in [x,y]$. If $i \in V$, then from (4) $\exists w' \in Conv(\hat{P}_i^x, \hat{P}_i^y) \subset \bigcup_i Conv(\hat{P}^x_i, \hat{P}^y_i)$ such that $\|w-w'\|_2 < \frac{\epsilon}{3}$. On the other hand, if $i \in U$, then by (2) $\exists w' \in P_j^z$ for some $j \in V$ such that $\|w-w'\|_2 < \frac{2\epsilon}{3}$. In turn by (1) again, $\exists w'' \in  Conv(\hat{P}_j^x, \hat{P}_j^y) \subset \bigcup_i Conv(\hat{P}^x_i, \hat{P}^y_i)$ such that $\|w'-w''\|_2 < \frac{\epsilon}{3}$. Using the triangle inequality $\|w-w''\|_2 < \epsilon$, and hence $\bigcup_i Conv(\hat{P}^x_i, \hat{P}^y_i)$ is an $\epsilon$-net of $(\Delta^m)^{x,y}$ in the $\ell_2$ norm as desired.

Let us prove statement (1). We know that if $i \in U$, for all $z\in [x,y]$ it holds that $\tau(P_i^z) \leq \alpha$. Using the union bound from Lemma \ref{lemma:subadd} we see $\tau( \cup_{i \in U} P_i^z ) \leq \frac{10n \alpha m^{3/2}}{3} = \frac{\epsilon}{6m}$.
On the other hand, we also know that $\cup_{i \in U} P_i^z \subset (\Delta^m)^z \cong (1-z) \Delta^{m-1}$. From Lemma \ref{lemma:thick-diam-simplex}, we know $\tau((\Delta^m)^z) \geq \frac{1-z}{((m-1) + \sqrt{m-1})}$. It follows that if $\frac{\epsilon}{6m} < \frac{1-z}{((m-1) + \sqrt{m-1})}$, then $\cup_{i \in U} P_i^z \neq (\Delta^m)^z$. The condition $y \leq 1 - \frac{\epsilon}{3}$ ensures that this happens for all $z \in [x,y]$. This in turn implies $V \neq \emptyset$.

Let us prove statement (2). From Lemma \ref{lemma:thick-diam-simplex}, we know $\frac{Diam((\Delta^m)^z))}{\tau((\Delta^m)^z)} \leq \sqrt{2}((m-1) + \sqrt{m-1})$. We can apply Lemma \ref{lemma:cones-nets} in exactly the same fashion as Lemma \ref{lemma:thickness-to-distance} to get $\gamma_1 = \frac{2\epsilon}{3} > \left(\frac{10n \alpha m^{3/2}}{3} \right) 4(m-1)$. We know that if $w \in \cup_{i \in U} P_i^z$, then $\exists w' \in \cup_{i \in V} P_i^z$ such that $\|w - w'\|_2 < \gamma_1 = \frac{2\epsilon}{3}$, which is what we wanted to show. 

Let us prove statement (3). Let us define $(\hat{P}_j^x)_\bot = P_j^x \setminus \left( \bigcup_{i \in [n]} \hat{P}_i^x \right)$. Note that these are the points in $P_j^x$ that do not have any label whatsoever under the empirical labelling at $x$. Importantly, some points could have a label other than $j$ if these points are on the boundary of another polytope with a label that has higher priority in the lexicographic ordering. By the fact that we have a $\beta$-close labelling of $\mathscr{P}^x$, it must hold that $\tau((\hat{P}_j^x)_\bot) \leq \beta$. Also, since $P_j^x \subset (\Delta^m)^x$, we know $Diam(P_j^x) \leq \sqrt{2}$ from Lemma \ref{lemma:thick-diam-simplex}. Since $j \in V$, we also know that $\tau(P_j^x) \geq \alpha$, hence $\tau((\hat{P}_j^x)_\bot ) \leq \beta < \alpha \leq \tau(P_j^x)$ which in turn implies that $int(P_j^x) \setminus  ( \hat{P}_j^x )_\bot \neq \emptyset$. Let $\eta^* = \frac{1}{2} \left( \frac{\epsilon}{3} - \left( \frac{\sqrt{2}}{\alpha} \right) \beta \right) > 0$ and let $\gamma_2 = \frac{\epsilon}{3} - \eta^* > \left(\frac{\sqrt{2}}{\alpha}\right) \beta$ (the addition of the $\eta^*$ gap is to help with the proof of statement (4)). We can use the second part of Lemma \ref{lemma:cones-nets} to see $B_{\gamma_2}(e) \cap \left( int(P_j^x) \setminus (\hat{P}_j^x)_\bot  \right) \neq \emptyset$. Since all points in $int(P_j^x)$ only belong to $P_j$, it follows that under the lexicographic oracle one only sees the label $j$ for those points. This implies that $B_{\gamma_2}(e) \cap \hat{P}_j^x \neq \emptyset$, which in turn implies $\exists e' \in \hat{P}_j^x$ such that $\|e-e'\|_2 < \gamma_2 = \frac{\epsilon}{3} - \eta^* < \frac{\epsilon}{3}$ as desired.

Finally, we prove statement (4). Since $[x,y]$ has no critical points, from Lemma \ref{lemma:perfect-fleshing} we know that $Conv(P_j^x,P_j^y) = P_j^{x,y}$, which in turn means that there exist $a \in P_j^x$ and $b \in P_j^y$ such that $w \in Conv(a,b)$. To be precise $w = Conv(a,b) \cap \pi^{-1}(z)$. Now if $a \in \hat{P}_j^x$ and $b \in \hat{P}_j^y$, then we are done. Let us suppose that this is not the case. We focus on $a$. If $a \in (\hat{P}_j^x)_\bot$, the previously proved statement says there is some $a' \in \hat{P}_j^x$ such that $\|a-a'\|<\frac{\epsilon}{3}$. If $a \notin (\hat{P}_j^x)_\bot \cup \hat{P}_j^x$, then it must be the case that $a \in \hat{P}_k^x \cap P_j^x$ for some other $k \in [n]$. This however only happens if $a \in P_j\cap P_k$ for some $P_k \neq P_j$, from the second property of polytope partitions from Definition \ref{def:polytope-partition} and the fact that using the lexicographic query oracle means that if $P_j = P_k$ and $j<k$ then $\hat{P}_k = \emptyset$ always. Invoking the second property of Definition \ref{def:polytope-partition} again, we see that $a$ lies on a bounding hyperplane of $P_j$. This in turn means that for every $\delta > 0$, $B_\delta(a) \cap int(P_j^x) \neq \emptyset$. Let us thus consider $\delta^* = \min\{ \frac{\epsilon}{3}, \frac{\eta^*}{2} \}$, where $\eta^*$ is defined as in the previous paragraph. Let $x^*$ be a point in $B_{\delta^*}(a) \cap int(P_j^x)$. Either $x^* \in \hat{P}_j^x$ or $x^* \in (\hat{P}_j^x)_{\bot}$. In the former case, since $\delta^* < \frac{\epsilon}{3}$ we are done, we have succeeded in finding $a' = x^* \in \hat{P}_j^x$ such that $\|a-a'\|_2 < \frac{\epsilon}{3}$. In the latter case, from the previous paragraph, since $x^* \in (\hat{P}_j^x)_{\bot}$ we know that $\exists a' \in \hat{P}_j^x$ such that $\|x^*-a'\|_2 < \frac{\epsilon}{3} - \eta^*$. Since $\delta^* < \frac{\eta^*}{2}$, we can use the triangle inequality to conclude that $\|a-a'\|_2 < \frac{\epsilon}{3}$. In either case, we have proven what we wanted.

The same argumentation as the previous paragraph shows us that $\exists b' \in \hat{P}_j^y$ such that $\|b-b'\|_2 < \frac{\epsilon}{3}$. If we let $w' = Conv(a',b') \cap \pi^{-1}(z)$, then $w'$ satisfies the requirements of statement (4) and we have finished our proof.
\end{proof}

For the following corollary, suppose that $\mathscr{P}$ is an $(m,n)$ polytope partition and that $0 = t_0 < t_1,...,<t_k = 1$ are points in $[0,1]$. Furthermore suppose that $\beta =  \frac{\epsilon^2}{85 n m^{5/2}}$ as in Lemma \ref{lemma:crux-for-gbs}. For each $t_i$, if $t_i \notin C^\alpha_{\mathscr{P}}$, let $\hat{\mathscr{P}}^{t_i}$ be a $\beta$-close labelling of $\mathscr{P}^x$, otherwise let $\hat{\mathscr{P}}^{t_i} = \emptyset$. Let $\hat{\mathscr{P}} = Conv_i(\hat{\mathscr{P}}^{t_i})$ and for $i = 1,..,k$, let $I_j = [t_{j-1}, t_j]$. If $\hat{\mathscr{P}}^{t_{i-1},t_i}$ is an $\epsilon$-close labelling of $\mathscr{P}^{t_{i-1},t_i}$, we say that $I_j$ is covered, otherwise we say $I_j$ is uncovered. 

\begin{corollary}
For any collection of $\{t_i\}_{i=1}^k$, there are no more than $2C^\alpha_{\mathscr{P}}$ intervals $I_j$ that are uncovered. 
\end{corollary}\label{cor:width-bound}

\begin{proof}
Suppose that $I_j$ is uncovered, then one of the following holds:

\begin{itemize}
    \item Either $t_{j-1}$ or $t_j$ are in $C^\alpha_{\mathscr{P}}$
    \item $t_{j-1}, t_j \notin C^\alpha_{\mathscr{P}}$ yet $Conv(\hat{P}^{t_{j-1}}, \hat{P}^{t_j})$ is not an $\epsilon$-close labelling of $\mathscr{P}^{t_{j-1}, t_j}$.       
\end{itemize}
From the contrapositive of Lemma \ref{lemma:crux-for-gbs}, the latter case implies $I_j \cap C^\alpha_{\mathscr{P}} \neq \emptyset$, hence in either case there is a critical coordinate in $I_j$. In the worst case each $x \in C^\alpha_{\mathscr{P}}$ lies on a $t_j$, causing both $I_j$ and $I_{j+1}$ to be uncovered. This implies that there are at most $2C^\alpha_{\mathscr{P}}$ intervals $I_j$ that are uncovered. 
\end{proof}

\subsection{Specification of CD-GBS and Query Usage}\label{sec:SBS}

\paragraph{Terms and Notation:} The details of CD-GBS are presented in Algorithm \ref{alg:GBS}. We recall our notation from Definition \ref{def:degenerate-cross-section} where for $x \in [0,1)$ we defined $f_x: (\Delta^m)^x \rightarrow \Delta^{m-1}$ given by $f_x(x,...,v_m) = \frac{1}{1-x}(v_2,...,v_m)$. We note that this is a bijection between both polytopes, hence it is well-defined to use $f_x^{-1}$. In addition, we let $\mathcal{D}^k = \{\frac{i}{2^k} \ | 1 \leq i \leq 2^i \}$ be the dyadic fractions of $k$-th power in the unit interval (excluding 0). For every $x \in \mathcal{D}^k$ we can associate the interval $I_x^k = [x - \frac{1}{2^k}, x]$. For each of these intervals $midpoint(I^k_x)$ denotes its midpoint. We also use the same language as Corollary \ref{cor:width-bound} when we talk about whether $I^k_x$ is covered or not (with respect to the current empirical labelling, $\hat{\mathscr{P}}$, obtained from taking convex hulls of labels in $\Delta^m$). We note that in order to have a well-defined base case of CD-GBS (which is equivalent to binary search), we let $\Delta^0 = \mathbb{R}^0 = \{0\}$. Finally, we say that a point $x \in [0,1]$ is an uncovered critical point if $\hat{\mathscr{P}}^x$ is computed via a recursive call to CD-GBS and for $(a,b) = B_{\epsilon/2}(x) \cap [0,1]$, it holds that $\hat{\mathscr{P}}^{a,b}$ is not an $\epsilon$-close labelling of $\mathscr{P}^{a,b}$.  

\begin{theorem}\label{thm:genbinsrch-correct}
If CD-GBS is given access to $Q_\ell$ for a $(m,n)$-polytope partition, it computes an $\epsilon$-close labelling of $\mathscr{P}$ using at most $\left( \prod_{i=1}^m \left( \binom{n+i}{i} + 2n \right) \right) 2^{2m^2}\log^m \left( \frac{170 n m^{5/2}}{\epsilon}  \right)$ membership queries. For constant $m$ this constitutes $O(n^{m^2}\log^m \left( \frac{n}{\epsilon}  \right) ) = poly(n,\log \left( \frac{1}{\epsilon} \right))$ queries \footnote{CD-GBS runs in polynomial time for constant $m$. The time-intensive operation consists of identifying uncovered intervals, but since the dimension of the ambient simplex is constant, each empirical polytope $\hat{P}_i$ has at most a constant number of bounding hyperplanes. These hyperplanes can each be extruded by $\epsilon$, and checking whether there exists a point outside all these extrusions can be done in time polynomial in $n$ via brute force. In fact, all other algorithms in this paper have efficient runtimes (in their relevant parameters) due to similar reasoning.}.
\end{theorem}

\begin{proof}
We first prove that CD-GBS indeed computes an $\epsilon$-close labelling when given access to a valid $Q_\ell$ by inducting on $m$. It is straightforward to see that in the case $m=1$, if CD-GBS is given access to a valid $Q_\ell$ for a $(1,n)$ polytope partition (a partition of the unit interval into conected subintervals), then it simply performs binary search on the interval $[0,1] \cong \Delta^1$. 

As for the inductive step, for $k = \lceil \log (2/\epsilon) \rceil$, any two contiguous points of $\mathcal{D}^k$ are less than $\epsilon/2$ away from each other. For now suppose that every recursive call to CD-GBS was along a non-degenerate cross section $\mathscr{P}^t$. From the inductive assumption, this means that CD-GBS computes an $\epsilon/2$-close labellings of those cross-sections, using the triangle inequality, we know that $\hat{\mathscr{P}}$ is an $\epsilon$-close labelling of $\mathscr{P}$.

We note however that there is no guarantee for what a recursive call to CD-GBS does on a degenerate cross section $\hat{\mathscr{P}}^t$. For this reason, it could be the case that at the end of the loop over $\mathcal{D}^i$, $\hat{\mathscr{P}}$ is not an $\epsilon$-close labelling. This can only happen if there is some $t \in C^\alpha_{\mathscr{P}} \cap \mathcal{D}^k$ which is an uncovered critical coordinate.      

If $t$ is an uncovered critical coordinate we can rectify the situation. If we find a $z \in B_{\epsilon/2}$ that is not a critical coordinate, then $\mathscr{P}^z$ is non-degenerate and computing CD-GBS along the cross-section gives us an $\frac{\epsilon}{2}$-close labelling of $\mathscr{P}^z$. Using the triangle inequality, we see that this in turn removes $t$ from the set of uncovered critical coordinates, and we say that $t$ is ``fixed''.  Thus the final while loop of the algorithm eliminates the set of uncovered critical coordinates so that $\hat{\mathscr{P}}$ is indeed an $\epsilon$-close labelling.

It thus remains to show that the final while loop terminates. However, there are at most $|C^\alpha_\mathscr{P}|$ uncovered critical coordinates, and over the course of fixing all uncovered critical coordinates, there are at most $|C^\alpha_\mathscr{P}|$ bad guesses for $z \in B_{\epsilon/2}(x)$ where $\mathscr{P}^z$ is degenerate. Therefore the final while loop makes at most $2|C^\alpha_\mathscr{P}|$ invocations to CD-GBS along cross-sections. This concludes the proof of correctness for CD-GBS.

Let us bound the total query usage of CD-GBS. For all values of $k$ in the first for loop, we know from Corollary \ref{cor:width-bound} that since $Q_\ell$ is a valid lexicographic oracle for $\mathscr{P}$, that the number of uncovered $I_x^k$ will not exceed $2 \left( \binom {n+m}{m} + 2n \right)$, and since CD-GBS is called once per uncovered interval, it follows that for each $k$ there at most $2 \left( \binom {n+m}{m} + 2n \right)$ recursive calls to CD-GBS. Furthermore, since $Q_\ell$ is a valid lexicographic oracle for $\mathscr{P}$, it will also never be the case that $\exists i,j \in [n], \ z \in \Delta^m$ such that $dim(\hat{P}_i) = m$ and $z \in int(\hat{P}_i)$.

In the worst case, $k$ loops from 1 to $\lceil \log(2/\epsilon) \rceil$ and makes an extra $2|C^\alpha_{\mathscr{P}}|$ recursive calls to CD-GBS to fix all uncovered critical coordinates. In total if we let $T(m,n,\epsilon)$ denote the query cost of running CD-GBS on a valid lexicographic oracle, we get the following recursion: 
$$
T(m,n,\epsilon) \leq 2 |C^\alpha_{\mathscr{P}}| \log \left( \frac{2}{\epsilon} \right) T \left(m-1, n, \frac{\epsilon^2}{85 n m^{5/2}} \right) + 2|C^\alpha_{\mathscr{P}}|
$$
In order to make this more amenable, we define $f(m) = \left(\binom{n+m}{m} + 2n \right)$ and use Lemma \ref{lemma:bound-critical-cardinality} to bound this expression as follows:
$$
T(m,n,\epsilon) \leq 3 f(m) \log \left( \frac{2}{\epsilon} \right) T \left(m-1, n, \frac{\epsilon^2}{85 n m^{5/2}} \right)
$$
Furthermore, from the fact that the base case is binary search, we know $T(1,n,\epsilon) \leq n \log \left( \frac{2}{\epsilon} \right)$.

To unpack the recursion. Let us define $\epsilon_0 = \epsilon$ and $\epsilon_{k+1} = \frac{\epsilon_k^2}{85n(m-k)^{5/2}}$ for $k = 1,...,m-1$. With this in hand, we can unroll the recursion to obtain:
$$
T(m,n,\epsilon) \leq  \left( 3^{m-1} \prod_{i=1}^{m-1} f(i) \right) 
\left( \prod_{k=1}^{m-1} \log \left( \frac{2}{\epsilon_k} \right) \right)
$$
Since each $\epsilon_{k+1} < \epsilon_k$, we can upper bound the right-hand product by bounding each term with $\epsilon_{m-1}$. If we first solve for this value, we obtain: 
$$
\epsilon_{m-1} = \frac{\epsilon^{2^{m-1}}}{\prod_{j=1}^{m-1}(85nj^{5/2})^{2^j}}
\geq \frac{\epsilon^{2^{m-1}}}{\prod_{j=1}^{m-1}(85nm^{5/2})^{2^j}}
\geq \left( \frac{\epsilon}{85nm^{5/2}} \right)^{2^m}.
 $$
In the first inequality we bounded the denominator product in the base by $j \leq m$, as for the second inequality, we evaluated the geometric series in 2 for the exponent to bound the exponent by $2^m$. With this in hand we obtain the desired bounds: 
$$
T(m,n,\epsilon) \leq  3^{m} 2^{m^2} \prod_{i=1}^{m} f(i)
\log^m \left( \frac{170nm^{5/2}}{\epsilon} \right)
\leq \left( \prod_{i=1}^m \left( \binom{n+i}{i} + 2n \right) \right) 2^{2m^2}\log^m \left( \frac{170 n m^{5/2}}{\epsilon}  \right)
$$

Finally, For large enough $n$, every term in $\prod_{i=1}^m \left( \binom{n+i}{i} + 2n \right)$ is bounded by $(n+m)^m +2n$. It follows that this product is $O(n^{m^2})$, and thus for constant $m$, this constitutes $O(n^{m^2}\log^m \left( \frac{n}{\epsilon}  \right) ) = poly(n,\log \left( \frac{1}{\epsilon} \right))$ queries.
\end{proof}

The previous results show that for constant dimension, $m$, CD-GBS is query efficient in $n$ and $\frac{1}{\epsilon}$. In the following section we use this algorithm as a building block to construct a method for computing efficient $\epsilon$-close labellings when the number of regions, $n$, is held constant instead. 

\begin{algorithm}[h] % enter the algorithm environment
\caption{CD-GBS$(m,n,\epsilon, Q)$} %ROUGH SKETCH} % give the algorithm a caption
\label{alg:GBS} % and a label for \ref{} commands later in the document
\begin{algorithmic} % enter the algorithmic environment
	\INPUT $m \geq 0, \  n,\epsilon > 0$, query access to function $Q: \Delta^m \rightarrow [n]$.
	\OUTPUT $\hat{\mathscr{P}}$: an $\epsilon$-close labelling of $\mathscr{P}$.
    %\item[\textbf{if} a \textbf{then} b]
    %\item[\textbf{if} $m=0$ \textbf{then} Query $Q(0)$]
    %\algorithmicifthenline{a}{b}
    \IF{$m=0$}
        \STATE Query $Q(0)$ 
    \ELSE
	\STATE $\hat{\mathscr{P}}^0 \leftarrow f_0^{-1} \left( \text{CD-GBS} \left( m-1,n,\frac{\epsilon^2}{85 n m^{5/2}}, Q \circ f_0^{-1} \right) \right)$, $\hat{\mathscr{P}}^1 \leftarrow Q(\vec{e}_1)$.
	%\STATE
	
	\FOR{$k = 1$ to $\lceil \log(2/\epsilon) \rceil$}
	    \IF{Number of uncovered $I_x^k$ exceeds $2 \left( \binom{n + m} {m} + 2n \right)$}
	        \STATE Halt
	    \ENDIF
	    \FOR{$x \in \mathscr{D}^k$}
	        \IF{$I_x^k$ is uncovered}
	            \STATE $t \leftarrow midpoint(I_x)$
	            \STATE $\hat{\mathscr{P}}^t \leftarrow f_t^{-1} \left( \text{CD-GBS} \left( m-1,n,\frac{\epsilon^2}{85(1-t) n m^{5/2}}, Q \circ f_t^{-1} \right) \right)$
	        \ENDIF
	        \STATE Recompute $\hat{\mathscr{P}}$ by taking convex hulls of labels
	        \IF{$\exists i,j \in [n]$ such that $int(\hat{P}_i) \cap \hat{P}_j \neq \emptyset$ or $\hat{\mathscr{P}}$ is an $\epsilon$-close labelling}
		        \STATE Halt
		    \ENDIF
	    \ENDFOR
	\ENDFOR
	\WHILE{$\exists x \in [0,1]$ an uncovered critical point }
	    \STATE $t \leftarrow z$ for arbitrary $z \in B_{\epsilon/2}(x)$
	    \STATE $\hat{\mathscr{P}}^t \leftarrow f_t^{-1} \left( \text{CD-GBS} \left( m-1,n,\frac{\epsilon^2}{85(1-t) n m^{5/2}}, Q \circ f_t^{-1} \right) \right)$
	    \STATE Recompute $\hat{\mathscr{P}}$ by taking convex hulls of labels
	\ENDWHILE
    \ENDIF
    \iffalse
	\STATE Define $\hat{Q}: \Delta^m \rightarrow [n] \cup \{\bot\}$ as follows:
	\FOR{$i \in [n]$}	
		\STATE if $ x \in \hat{P}_i$,  $\hat{Q}(x) \leftarrow i$, else $\hat{Q}(x) = \bot$
	\ENDFOR
	\fi
	\RETURN $\hat{\mathscr{P}}$    
\end{algorithmic}
\end{algorithm}

\section{Constant-Region Generalised Binary Search for \texorpdfstring{$Q_\ell$}{}}\label{sec:GBS}

In this section we introduce Constant-Region Generalised Binary Search, (CR-GBS), which as the name suggests, is a query-efficient algorithm for computing $\epsilon$-close labellings of $(m,n)$-polytope partitions when $n$ is constant and $m$ and $\epsilon$ are allowed to vary. 

The intuition behind the algorithm lies in the fact that if $m$ is much greater than $n$ (it suffices for $m > \binom {n} {2}$ ), then any vertex of a given $P_i$ cannot lie in the interior of the ambient simplex $\Delta^m$. This is because a vertex in $\Delta^m$ must consist of the intersections of at least $m$ half-spaces, all of which cannot arise from adjacencies between different $P_i$.

Not only do all vertices lie on the boundary of $\Delta^m$, but one can easily show that they are all contained in faces of the boundary of $\Delta^m$ that have dimension $O(n^2)$ which is presumed to be constant. The number of such faces in the boundary of $\Delta^m$ is thus polynomial in $m$, and moreover if we could compute $0$-close labellings of these faces we could take convex combinations and recover a $0$-close labelling of the entire polytope partition. 

We will demonstrate that for an appropriate value of $\epsilon'$, if we compute $\epsilon'$-close labellings of such faces in the boundary, we can recover an $\epsilon$-close labelling of the entire polytope partition over all of $\Delta^m$ by taking convex combinations. CR-GBS computes the necessary $\epsilon'$-close labellings of lower dimensional faces by using CD-GBS as a subroutine, which as we shall see results in our desired query efficiency for $n$ constant.

We note however, that not all faces in the boundary of $\Delta^m$ are axis-aligned, which poses a problem if we are to use CD-GBS as a subroutine. As we show in the following section, this is not an issue since we can translate such simplices into axis-aligned simplices via a simple transformation.

\subsection{Non-axis-aligned Simplices} \label{subsec:face-to-simplex}
So far we have focused on the case where $\Delta^m = \{x \in \mathbb{R}^m \ | \ \|x\|_1 \leq 1,$ $x_i \geq 0\}$. In a straightforward fashion we transform our results to the equivalent simplex $\Lambda^m = \{x \in \mathbb{R}^{m+1} \ | \ \|x\|_1 = 1$, $x_i \geq 0\}$. To do so, we define the invertible linear map $\phi_m: \Delta^m \rightarrow \Lambda^{m+1}$ given by $\phi_m(x_1,...,x_m) = \left( \left( 1-\sum_{i=1}^m{x_i} \right), x_1,...,x_m \right)$. It is straightforward to see that $\phi_m$ is $\sqrt{m+1}$-Lipschitz continuous. Via standard Lipschitz continuity arguments we get the following:

\begin{lemma}
Suppose that $\mathscr{P}$ is an $(m,n)$-polytope partition of $\Lambda^{m+1}$. If $\hat{\mathscr{P}}$ is an $\frac{\epsilon}{\sqrt{m+1}}$-close labelling of $\phi^{-1}_m(\mathscr{P})$, then $\phi_m (\hat{\mathscr{P}})$ is an $\epsilon$-close labelling of $\mathscr{P}$.
\end{lemma}

\begin{proof}
Suppose that $x\in \Lambda^{m+1}$ has no label under $\phi_m(\hat{\mathscr{P}})$. Since $\hat{\mathscr{P}}$ is an $\frac{\epsilon}{\sqrt{m+1}}$-close labelling, there must be some $y \in \Delta^m$ with the property that $\|\phi_m^{-1}(x) - y\|_2 < \frac{\epsilon}{\sqrt{m+1}}$. If we consider $\phi_m^{-1}(x)$, since $\hat{\mathscr{P}}$ is $\frac{\epsilon}{\sqrt{m+1}}$-close, there must be some $y$ from an empirical polytope $\hat{P}_i \subset \Delta^m$ with the property that $\|\phi_m(x) - y\|_2 < \frac{\epsilon}{\sqrt{m+1}}$. It follows that $\phi_m(y) \in \phi_m(\hat{\mathscr{P}})$ and by Lipschitz continuity of $\phi_m$, $\|x - \phi_m(y)\|_2 < \epsilon$ as desired.
\end{proof}

\subsection{Necessary Machinery for CR-GBS}
\label{subsec:overview-CR-GBS}

Suppose that $\mathscr{P}$ is an $(m,n)$-polytope partition with the property that $m > \binom{n}{2}$. Furthermore, let $k = \binom{n}{2}$ and let $\partial_k(\Delta^m)$ denote all $k$-dimensional faces of $\Delta^m$. For each face $F$, let $\mathscr{P}_F$ be the restriction of $\mathscr{P}$ to $F$. If $F$ is axis-aligned (equivalently, if $F$ contains the origin), then it is an isometric embedding of $\Delta^k$ in $\Delta^m$, so we let $\phi_F$ be a canonical isomorphism from $F$ to $\Delta^k$. If $F$ is not axis-aligned, we let $\phi_F$ be any canonical isomorphism from $F$ to $\Delta^k$ as per Section \ref{subsec:face-to-simplex}.

As mentioned previously, computing empirical labellings of every face in $\partial_k(\Delta^m)$ via CD-GBS will be enough to compute an empirical labelling for $\mathscr{P}$. The only issue with this strategy however, is that CD-GBS is only guaranteed to return an $\epsilon$-close labelling if it is given access to a valid lexicographic membership oracle for a polytope partition, and for an arbitrary polytope partition, it is not always the case that $\phi_F(\mathscr{P}_F)$ is a $(k,n)$-polytope partition for all $F \in \partial_k(\Delta^m)$. As an example, consider a polytope partition with an arbitrary $m-1$-dimensional polytope $P_i$ contained in $F = \mathscr{P}^0$ (the $0$-cross-section of $\mathscr{P}$). Any full-dimensional $P_j \in \mathscr{P}$ must have the property that $0 \notin \pi(relint(P_j))$, hence it still holds that $relint(P_i) \cap relint(P_j) = \emptyset$. However, when restricted to $\mathscr{P}_F$, relative interiors are with respect to $\mathscr{P}^0$, and it can be the case that $relint((P_i)_F) \cap relint((P_j)_F) \neq \emptyset$. For this reason, we slightly refine our notion of polytope partition.

\begin{definition}\label{def:proper-polytope-partition}
Suppose that $\mathscr{P}$ is an $(m,n)$-polytope partition such that for all $0 \leq k \leq m$ and $F \in \partial_k(\Delta^m)$, $\phi_F(\mathscr{P}_F)$ is a $(k,n)$-polytope partition. Then we say that $\mathscr{P}$ is a {\em proper polytope partition}. 
\end{definition}

For the remainder of this section, we focus on proper polytope partitions. In addition, in order to prove correctness of CR-GBS we define a robust approximation of any $P_i \in \mathscr{P}$. 

\begin{definition}\label{def:robust-clipping}
Suppose that $P \subset \Delta^m$ is a polytope. We define $int_\gamma(P)$ as
$$
int_\gamma(P) = \{x \in P \ | \ B_\gamma(x) \cap \Delta^m \subset P \}
$$
We call this the {\em $\gamma$-interior} of $P$.
\end{definition}

Intuitively, the $\gamma$-interior of $P$ consists of points that are ``robustly'' within $P$ by a margin of $\gamma$ relative to the interior of $\Delta^m$, as visualised in Figure \ref{fig:gamma-interior}. In Lemma~\ref{lemma:robust-vertex-faces} we show that $int_\gamma(P)$ is a sub-polytope of $P$ with certain supporting hyperplanes translated towards the interior of $P$ by a margin of $\gamma$. We also show that if $P_i$ is an element of an $(m,n)$-polytope partition where $m > k$, then the vertices of $int_\gamma(P_i)$ also lie in some $F \in \partial_k(\Delta^m)$.

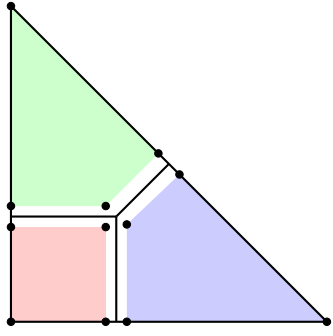
\begin{figure}
\center{
\begin{tikzpicture}[scale=0.7, every node/.style={draw,shape=circle,fill=blue}]
\tikzstyle{xxx}=[dashed,thick]

%\draw[cyan!20] (0,0) grid (12,10);

% Left Partition colors
%\fill[red!20](-4.25,2)--(-4.25,3.8)--(-6,5.6)--(-6,2)--cycle; %red left
%\fill[green!20](-6,6.4)--(-4,4.2)--(-2.2,4.2)--(-6,8)--cycle; %green left
%\fill[blue!20](-3.75,2)--(-3.9,2.25)--(-3.75,3.8)--(-2,3.9)--(-0.5,2.5)--(-0.5,2)--cycle; blue left

%Right Partition Colors
\fill[red!20](2,2)--(2,3.8)--(3.8,3.8)--(3.8,2)--cycle; %red right
\fill[blue!20](4.2,2)--(8,2)--(5.2,4.8)--(4.2,3.85)--cycle; %blue right
\fill[green!20](2,4.2)--(3.8,4.2)--(4.8,5.2)--(2,8)--cycle; %green right 

%x-axes
\draw[thick, -](2,2)--(8,2);

%Left y-axes
\draw[thick, -](2,2)--(2,8);

%Triangle outlines
\draw[thick, -](2,8)--(8,2);

%polytope boundaries
\draw[thick, -](4,4)--(5,5);
\draw[thick, -](4,4)--(4,2);
\draw[thick, -](2,4)--(4,4);

%Right sample nodes red
\filldraw (2,2) circle (2pt);
\filldraw (2,3.8) circle (2pt);
\filldraw (3.8,3.8) circle (2pt);
\filldraw (3.8,2) circle (2pt);

%Right sample nodes blue
\filldraw (4.2,2) circle (2pt);
\filldraw (8,2) circle (2pt);
\filldraw (5.2,4.8) circle (2pt);
\filldraw (4.2,3.85) circle (2pt);

%Right sample nodes green
\filldraw (2,4.2) circle (2pt);
\filldraw (3.8,4.2) circle (2pt);
\filldraw (4.8,5.2) circle (2pt);
\filldraw (2,8) circle (2pt);

\end{tikzpicture}
\caption{$\gamma$-Interiors of Polytopes in a Partition }\label{fig:gamma-interior}
}
\end{figure}

\begin{lemma}\label{lemma:robust-vertex-faces}
Suppose that $\mathscr{P}$ is an $(m,n)$-polytope partition with $m > k = \binom{n}{2}$. For each $P_i \in \mathscr{P}$, and any $\gamma > 0$, $int_\gamma(P_i)$ is a sub-polytope of $P_i$. Furthermore, each vertex of $int_\gamma(P_i)$ lies in some $F \in \partial_k(\Delta^m)$. 
\end{lemma}

\begin{proof}
Since $P_i \subset \mathbb{R}^m$ is a polytope, it can be expressed as the intersection of finitely many half-spaces: $P_i = \bigcap_{j=1}^q H_j$, such that $H_j = \{x \in \mathbb{R}^m \ | \ a_j \cdot x \geq b_j, \text{ where } a_j \in \mathbb{R}^m, \ \|a_j\|_2 = 1, \ b_j \in \mathbb{R}\}$. As mentioned before, each half-space, $H_j$, can either arise as an adjacency of $P_i$ with the boundary of $\Delta^m$, or as an adjacency of $P_i$ with some other $P_r \in \mathscr{P}$. Let us call the former set of half-spaces $A$ and the latter $B$. We abuse notation slightly and also let $A$ refer to the sets of indices $j \in [q]$ such that $H_j \in A$ (similarly for $B$).

For each $H_j \in B$, let $H_j' = \{x \in \mathbb{R}^m \ | \ a_j \cdot x \geq b_j + \gamma, \text{ where } a_j \in \mathbb{R}^m, \ \|a_j\|_2 = 1, \ b_j \in \mathbb{R}\}$. Clearly $H_j' \subset H_j$, and in fact the boundary hyperplane of $H_j'$ is parallel to that of $H_j$ (and translated by a margin of $\gamma$ towards the interior of $H_j$). We now define $C = \left( \bigcap_{j \in A} H_j \right) \cap \left( \bigcap_{j \in B} H_j' \right)$ and we show that $int_\gamma(P_i) = C$, which proves the first part of the lemma. 

Suppose that $x \in C$. By virtue of the construction of all $H_j'$, it must be the case that $B_\gamma(x)$ does not intersect the boundary of any $H_j \in B$. Since all $H_j \in A$ are unchanged in $C$, we obtain $B_\gamma(x) \cap \Delta^m \subset P_i$, therefore $x \in int_\gamma(P_i)$. 

Now suppose that $x \in int_\gamma(P_i)$. Since $int_\gamma(P_i) \subset P_i$, it is clear that $x \in H_j$ for all $H_j \in A$. As for $H_j \in B$, we know that $x \in H_j$ from the fact that  $int_\gamma(P_i) \subset P_i$. If $x \notin H_j'$, then $B_\gamma(x) \not \subset H_j$, which in turn implies $B_\gamma(x) \not \subset P$, contradicting our assumption that $x \in int_\gamma(P_i)$. This proves the claim that $C = int_\gamma(P_i)$.

As for the final claim of the lemma, we note that since each $H_j \in A$ arises as an adjacency of $P_i$ with the boundary of $\Delta^m$, it must be the case that $|A| \leq m$. Furthermore, since each $H_j \in B$ arises as an adjacency of two polytopes in $\mathscr{P}$, it follows that $|B| \leq \binom {n} {2} = k < m$. Since at least $m$ half-spaces need to meet in $\mathbb{R}^m$ to make a vertex, it must be the case that any vertex of $C = int_\gamma(P_i)$ lies on some $F \in \partial_k(\Delta^m)$.
\end{proof}

Suppose that $\mathscr{P}$ is a proper $(m,n)$-polytope partition with $m > k = \binom{n}{2}$. Furthermore, suppose that $P_i \in \mathscr{P}$ is of full affine dimension and consider a vertex, $v$, of $int_\gamma(P_i)$ which is ``robustly'' in the interior of $P_i$ by definition. From the previous lemma we know that $v$ lies in some $F \in \partial_k(\Delta^m)$. We now show that due to the margin $\gamma$ with which $v$ lies within $P_i$, we can recover a label of $v$ by computing a suitable empirical-labelling of $F$.

\begin{lemma} \label{lemma:face-to-thickness}
Suppose that $\mathscr{P}$ is a proper $(m,n)$-polytope partition with $m > k = \binom{n}{2}$. Furthermore, suppose that $P_i \in \mathscr{P}$ is of full affine dimension and that $v$ is a vertex of $int_\gamma(P_i)$ that lies on some face $F \in \partial_k (\Delta^m)$. It follows that any $\frac{2\gamma}{5}$-close labelling of $F$ that correctly labels the vertices of $F$ gives $v$ the label $i$. Furthermore, suppose that for all $F \in \partial_k(\Delta^m)$ we compute a $\frac{2\gamma}{5}$-close labelling. By taking convex combinations of these empirical labellings, we get $\tau(\hat{P}_\bot) \leq \frac{10}{3} n^2 \gamma (m+1)^{3/2}$.
\end{lemma}

\begin{proof}
%To clarify notation, we let $B^1_r(x)$ and $B^2_r(x)$ denote $\ell_1$ and $\ell_2$ balls around $x$ respectively. 
If $v$ is a vertex as in the statement of the lemma, it must either be a vertex of the original simplex, or $B_\gamma(v) \cap P_i\cap F$ must contain a $r$-dimensional $\ell_2$ ball of radius $\gamma$  which we call $A_2$ (where $r$ corresponds to the dimension of the sub-face of $F$ that $v$ lies on, implying $1 \leq r \leq k$). If $v$ is a vertex of the original simplex, then it is correctly labelled by assumption, so we focus on the the latter case.

Let $A_1$ be any $r$-dimensional $\ell_1$ ball of radius $\frac{3\gamma}{5}$ such that $A_1 \subsetneq A_2$, and denote the corners of $A_1$ by $x_1,...,x_s$. For $i = 1,...,s$, let $V_i = B_{2\gamma/5}(x_i) \cap A_2 \subset F$. We note that $V_i \cap V_j = \emptyset$ for all $i \neq j$.

By the conditions of empirical labellings and the fact that $P_i$ is of full affine dimension, there must exist $z_1,...,z_s$ such that $z_r \in V_r$ and $z_r$ gets its correct label, $i$, under $Q_\ell$. Furthermore, it is straightforward to see that $v \in Conv(z_1,...,z_s)$, hence $v$ gets its correct label, $i$, as visualised in Figure \ref{fig:face-thickness} for $r = 2$. 

Along with Lemma \ref{lemma:robust-vertex-faces}, this shows that if for all $F \in \partial (\Delta^m)^k$ we compute $\frac{2\gamma}{5}$-close labellings that correctly label the vertices of $F$, then we will have correctly labelled all vertices of the polytope $int_\gamma(P_i)$. Consequently, by taking convex combinations of these labellings, the entirety of $int_\gamma(P_i)$ will be labelled correctly for an arbitrary full-dimensional $P_i$.

For a given full-dimensional $P_i \subset \mathscr{P}$, it is the case that $P_i \setminus int_\gamma(P_i)$ can be expressed as a disjoint union of at most $k \leq n^2$ polytopes of thickness bounded by $\gamma$ (using the notation from the proof of Lemma \ref{lemma:robust-vertex-faces}, these polytopes are all of the form $(H_j \setminus H_j') \cap P_i$, of which there are  at most $|B| = k$ ). For a given $P_j$ that is not full-dimensional, it trivially holds that $\tau(P_j) = 0$ Thus we can use Lemma \ref{lemma:subadd} to see that $\tau(\hat{P}_\bot) = \tau(\cup_i \left(P_i \setminus int_\gamma(P_i)\right) ) \leq \frac{10}{3} n^2 \gamma (m+1)^{3/2} $. 
\end{proof}

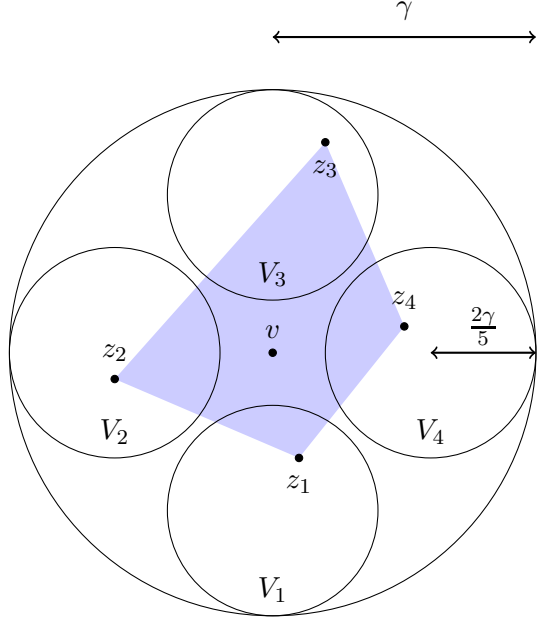
\begin{figure}
\center{
\begin{tikzpicture}[scale=0.7]
\tikzstyle{xxx}=[dashed,thick]

%\draw[cyan!20] (0,0) grid (12,10);

\fill[blue!20](-3,-0.5)--(0.5,-2)--(2.5,0.5)--(1,4)--cycle; %P4 bottom top right

%central point
\filldraw (0,0) circle (2pt);
%central circle
\draw (0,0) circle (5cm);

\iffalse
%limits of important balls
\filldraw (0,3) circle (2pt);
\filldraw (0,-3) circle (2pt);
\filldraw (3,0) circle (2pt);
\filldraw (-3,0) circle (2pt);
\fi

%samples from balls
\filldraw (1,4) circle (2pt);
\filldraw (0.5,-2) circle (2pt);
\filldraw (2.5,0.5) circle (2pt);
\filldraw (-3,-0.5) circle (2pt);

%important balls
\draw (0,3) circle (2cm);
\draw (0,-3) circle (2cm);
\draw (3,0) circle (2cm);
\draw (-3,0) circle (2cm);

%labelling radii
\draw[thick, <->](3,0)--(5,0);
\draw[thick, <->](0,6)--(5,6);

%labelling lengths
\node at (2.5,6.5) {$\gamma$};
\node at (4,0.5) {$\frac{2\gamma}{5}$};

%labelling vertex
\node at (0,0.4) {$v$};

%labelling inner circles
\node at (0,-4.5) {$V_1$};
\node at (-3,-1.5) {$V_2$};
\node at (0,1.5) {$V_3$};
\node at (3,-1.5) {$V_4$};

%labelling queried points
\node at (0.5,-2.5) {$z_1$};
\node at (-3,0) {$z_2$};
\node at (1,3.5) {$z_3$};
\node at (2.5,1) {$z_4$};

\end{tikzpicture}
\caption{Proof of Lemma \ref{lemma:face-to-thickness}}
\label{fig:face-thickness}
}
\end{figure}

\begin{corollary}\label{ref:corollary-for-gbs-correct}
Suppose that $\mathscr{P}$ is a proper $(m,n)$-polytope partition. Let $\gamma = \frac{3\epsilon}{40n^2(m+1)^{5/2}}$, and suppose that for all $F \in \partial_k(\Delta^m)$, a $\frac{2\gamma}{5}$-close labelling that correctly labels the vertices of $F$ is computed with $Q_\ell$. Taking a convex combination of these empirical labellings results in an $\epsilon$-close labelling of $\mathscr{P}$.
\end{corollary}

\begin{proof}
This follows from the fact that $\tau(\hat{P}_\bot) \leq \frac{10}{3} n^2 \gamma (m+1)^{3/2}$ from the previous theorem. We can therefore use Lemma \ref{lemma:thickness-to-distance} and obtain the desired result.
\end{proof}

The previous result gives us precisely what we need to prove the correctness of CR-GBS. In fact, it shows that CR-GBS can use any algorithm as a sub-routine (not just CD-GBS) as long as it computes empirical labellings of polytope partitions along all faces $F \in \partial (\Delta^m)^k$ while correctly labelling the vertices of $\Delta^m$.

\subsection{Specification of CR-GBS and Query Usage}

\paragraph{Terms and Notation:} For $F \in \partial_k(\Delta^m)$, we let $\phi_F$ denote a canonical isomorphism from $F$ to $\Delta^k$ as per Section \ref{subsec:overview-CR-GBS}. Furthermore, for each such $F$, we let $\hat{\mathscr{P}}_F$ empirical labelling returned by CD-GBS on a given face, $F$. 

\begin{algorithm} % enter the algorithm environment
\caption{CR-GBS$(m,n,\epsilon, Q)$} % give the algorithm a caption
\label{alg:FGBS} % and a label for \ref{} commands later in the document
\begin{algorithmic} % enter the algorithmic environment
	\INPUT $m,n,\epsilon > 0$, query access to membership oracle $Q$ for $(m,n)$-polytope partition $\mathscr{P}$.
	\OUTPUT $\epsilon$-close labelling of $\mathscr{P}$.
	\STATE $k \leftarrow \binom{n} {2}$
	\FOR{$F \in \partial_k(\Delta^m)$}
        \STATE $\hat{\mathscr{P}}^F \leftarrow \phi_F^{-1} \left( \text{CD-GBS}\left(k, n, \frac{3\epsilon}{100n^2\sqrt{k+1}(m+1)^{5/2}}, Q \circ \phi^{-1}_F \right) \right)$.
	\ENDFOR
	\STATE $\hat{\mathscr{P}} \leftarrow Conv_F(\hat{\mathscr{P}}_F)$
	\iffalse
	\STATE Define $\hat{Q}: \Delta^m \rightarrow [n] \cup \{\bot\}$ as follows:
	\FOR{$i \in [n]$}	
		\STATE if $ x \in \hat{P}_i$,  $\hat{Q}(x) \leftarrow i$, else $\hat{Q}(x) = \bot$
	\ENDFOR
	\fi
	\RETURN $\hat{\mathscr{P}}$  
\end{algorithmic}
\end{algorithm}

\begin{theorem}\label{thm:full-gbs-guarantee}
Let $\mathscr{P}$ be a proper $(m,n)$-polytope partition where $n$ is constant and $m > k = \binom{n}{2}$. CR-GBS computes an $\epsilon$-close labelling of $\mathscr{P}$ and uses $O \left( m^k \log^k \left( \frac{m}{\epsilon} \right) \right) = poly(m,\log \left( \frac{1}{\epsilon} \right) )$ queries.
\end{theorem}

\begin{proof}
The correctness follows from Corollary \ref{ref:corollary-for-gbs-correct}. In the worst case faces are of the form $\Lambda^k$, which incur an extra cost of $\sqrt{k+1}$ in the approximation factor of empirical labellings. We use this as a worst case bound.

For simplicity in notation, we define $m_0 = k$, $\epsilon_0 = \frac{3\epsilon}{100n^2\sqrt{k+1}(m+1)^{5/2}}$. From Theorem  \ref{thm:genbinsrch-correct}, the CD-GBS subroutine uses at most $\left( \prod_{i=1}^{m_0} \left( \binom{n+i}{i} + 2n \right) \right) 2^{2m_0^2}\log^{m_0} \left( \frac{170 n {m_0}^{5/2}}{\epsilon_0}  \right)$ queries. Since $k = \binom{n}{2}$ is constant, this expression can be written as $O \left( \log^k \left( \frac{m^{5/2}}{\epsilon} \right) \right) = O \left( \log^k \left( \frac{m}{\epsilon} \right) \right)$. Finally, there are $\binom{m}{k}$ possible faces upon which CD-GBS can be called as a subroutine, hence the total query usage is indeed $O \left( m^k \log^k \left( \frac{m}{\epsilon} \right) \right) = poly(m,\log \left( \frac{1}{\epsilon} \right) )$.
\end{proof}

\section{Upper Envelope Polytope Partitions}\label{sec:uepp}

Up until now we have focused completely on the lexicographic query oracle $Q_\ell$, creating algorithms CD-GBS and CR-GBS that compute $\epsilon$-close labellings of $(m,n)$-polytope partitions when given access to $Q_\ell$. If these algorithms are given access to an adversarial oracle $Q_A$ however, they may fail. It suffices to see this for CD-GBS since CR-GBS uses it as a subroutine. 

To see why CD-GBS may fail under $Q_A$ we recall that the algorithm recursively computes $\epsilon$-close labellings of cross-sections $\mathscr{P}^t$ for different values of $t \in [0,1]$. If ever CD-GBS is called on a degenerate cross-section $\mathscr{P}^t$, it has conditions to either tell that it is being called on a degenerate cross-section (when it notices that there exist $i,j \in [n]$ and $z \in \Delta^m$ such that $z \in int(\hat{P}_i) \cap \hat{P}_j$), or in the worst case, prevent it from exceeding its query balance. In both cases however, the algorithm returns a valid empirical labelling, i.e., $\hat{P} = \{ \hat{P}_i\}_{i=1}^n$ such that $\hat{P}_i \subseteq P_i$. 

When an adversarial oracle is used however, we may see $i,j \in [n]$ and $z \in \Delta^m$ such that $z \in int(\hat{P}_i) \cap \hat{P}_j$. Indeed this can occur if $P_i = P_j$ and both are full-dimensional. The natural solution seems to merge $P_i$ and $P_j$ (since the second condition of the definition of polytope partitions tells us that $P_i = P_j$ in this case). The main problem however, is that there is no way of telling when the condition above is an artifice of the adversarial oracle, or simply due to the fact that $\mathscr{P}^t$ is degenerate. If we blindly merge labels, we may in fact be performing an incorrect merge on a degenerate cross-section! This of course may return inconsistent polytope partitions. 

Since the key problem is the existence of degenerate cross-sections, we consider a slightly stronger variant of polytope partitions with the key property that cross-sections are never degenerate. Furthermore, this special type of polytope partition is expressive enough for our game theoretic applications, and best of all, it allows us to prove results in the adversarial query oracle model.

\begin{definition}[Upper Envelope Polytope Partition]\label{def:UE-polytope-partition}
Suppose that $A \in \mathbb{R}^{n \times m}$ is an $n \times m$ real-valued matrix and that $b \in \mathbb{R}^n$. Let $P_i = y \in \Delta^m$ such that $(Ay + b)_i \geq (Ay + b)_j$ for all $j \neq i$. We denote the collection $\mathscr{P}(A,b) = P_1,\ldots,P_n$, as the upper envelope polytope partition (UEPP) arising from $(A,b)$.
\end{definition}

It is straightforward to see that for any $(A,b)$, $\mathscr{P}(A,b)$ is itself an $(m,n)$-polytope partition. Crucially however, it satisfies more properties than the previous definition of polytope partitions.

\begin{lemma}\label{lemma:structure-uepp}
Suppose that $A$ is an $n \times m$ real valued matrix and that $b \in \mathbb{R}^n$. Then $\mathscr{P}(A,b) = \{P_1,\ldots,P_n\}$ has the following properties:
\begin{itemize}
    \item For any $x \in [0,1)$ let $f_x$ be the canonical affine transformation that maps $(\Delta^m)^x$ to $\Delta^{m-1}$. There exists an $n \times (m-1)$ real matrix $A^x$ and $b^x \in \mathbb{R}^n$ such that $\mathscr{P}(A^x,b^x) = f_x(\mathscr{P}(A,b)^x)$.  
    \item $\mathscr{P}(A,b)$ is a proper polytope partition (Definition~\ref{def:proper-polytope-partition}).
    \item If $A_{i,\bullet} = A_{j,\bullet}$ and $b_i = b_j$ then $P_i = P_j$. Conversely if $P_i$ is of full affine dimension and $relint(P_i) \cap P_j \neq \emptyset$, then $A_{i,\bullet} = A_{j,\bullet}$ and $b_i = b_j$; consequently, $P_i = P_j$.  
    \item Suppose that $a_1,\ldots,a_k \in \mathbb{R}$ are such that $\sum_{i=1}^k a_i < 1$ with $k < m$. Let $H = \{ (z_1,\ldots,z_m) \in \Delta^m \ | \ z_i = a_i, i=1,\ldots,k\}$ where $H$ has affine codimension $k$. If $x_1,\ldots,x_{m-k} \in \Delta^m$ are affinely independent points of  $P_i\cap H$ and $y \in Conv(x_1,\ldots,x_{m-k})$ belongs to $P_j$, then $P_i$ and $P_j$ coincide in $H$.   
\end{itemize}
\end{lemma}

\begin{proof}
The first bullet point follows from two facts: affine transformations restricted to affine subspaces are themselves affine transformations, and compositions of affine transformations are themselves affine transformations. 

To be rigorous, define the affine transformation $g:\mathbb{R}^m \rightarrow \mathbb{R}^m$ to be $g(x) = Ax + b$. Let $g' = g \restriction_{(\Delta^m)^x}$ be the restriction of $g$ to the affine subspace $(\Delta^m)^x \subset \mathbb{R}^m$ of codimension 1. As we mentioned before, $g'$ is itself an affine transformation. 

Now let us recall that $f_x$ is the canonical affine transformation that maps $(\Delta^m)^x$ to $\Delta^{m-1}$. It is straightforward to see that $f_x^{-1}$ exists ($\Delta^{m-1}$ and $(\Delta^{m})^x$ are clearly isomorphic) and is itself an affine transformation. Consequently $g' \circ  f_x^{-1} : \Delta^{m-1} \rightarrow \mathbb{R}^n$ is itself an affine transformation, which can be identified with a matrix $A^x$ and vector $b^x$ such that $\left( g' \circ  f_x^{-1} \right) z = A^x z + b^x$. It is straightforward to see that $(A^x,b^x)$ are such that $\mathscr{P}(A^x,b^x) = f_x(\mathscr{P}(A,b)^x)$ as desired.  

As for the second bullet point, let $F \in \partial_{m-1}(\Delta^m)$ be an arbitrary face of $\Delta^m$ of codimension 1. In addition, we use the notation $\phi_F$ as before to denote the canonical isomorphism from $F$ to $\Delta^{m-1}$. $\phi_F$ is itself an affine transformation, hence we can use identical argumentation from before to show that by restricting the original affine functions arising from $(A,b)$ to $F \in \partial_{m-1}(\Delta^m)$, we can find equivalent affine functions that render $\phi_F(\mathscr{P}_F)$ an upper-envelope polytope partition. For arbitrary $0 \leq k \leq m-1$, we can use the previous statement inductively to show that for any $F \in \partial_k(\Delta^m)$, $\mathscr{P}_F$ is a $(k,n)$-polytope partition. This concludes the proof that $\mathscr{P}$ is a proper polytope partition.

As for the third bullet point, the fact that $A_{i,\bullet} = A_{j,\bullet}$ and $b_i = b_j$ implies $P_i = P_j$ is trivial. Let us focus on the case when $P_i$ is of full affine dimension and $relint(P_i) \cap P_j \neq \emptyset$. For the sake of contradiction, let us suppose that $A_{i,\bullet} = A_{j,\bullet}$ and $b_i \neq b_j$. If this holds, then $(Ay + b)_i \neq (Ay + b)_j$ for all $y$, which contradicts our assumption that $relint(P_i) \cap P_j \neq \emptyset$. Let us therefore suppose that $A_{i,\bullet} \neq A_{j,\bullet}$. Let $H$ be the set of $y$ such that $(Ay + b)_i = (Ay + b)_j$. Since $A_{i,\bullet} \neq A_{j,\bullet}$, $H$ has codimension of at least 1. By assumption, there exists a $z \in relint(P_i) \cap P_j$. It must be the case that $z \in H$ as well. However, using the fact that $z \in relint(P_i)$ and that $P_i$ is of full affine dimension, for some $\epsilon> 0$, the $B_\epsilon(z) \subsetneq P_i$. However, since $z \in H$, which is of codimension 1, then half of $B_\epsilon(z)$ must not belong to $P_i$, which is a contradiction.  

The final bullet point follows from putting the first and third bullet points together and inducting on $k$. The base case follows from the fact that for $w \in [0,1)$, we know that $\mathscr{P}^w$ is itself a scaled upper envelope polytope partition (from the first bullet point). Now suppose that $x_1,...,x_{m-1} \in P_i$ are affinely independent in $\mathscr{P}(A,b)^w$. Furthermore suppose that $Conv(x_1,...,x_{m-1})$ contains a point $y \in P_j$. Since the $x_i$ are affinely independent, it follows that $P_i^w$ is full-dimensional in $\mathscr{P}(A,b)^w$, hence we can apply the third bullet point to show that $P_i^w$ and $P_j^w$ coincide in $\mathscr{P}(A,b)^w$ (which is in fact what we desired). 

Let us suppose that the claim holds for a given $k - 1 < m-1 $ and that we are given $a_1,...,a_k$. From the first bullet point, $\mathscr{P}(A,b)^{a_1}$ is a scaled lower-dimensional upper envelope polytope partition. Let us define $H = \{ (z_1,...,z_m) \in \Delta^m \ | \ z_i = a_i, i=1,..,k\}$ and $H_2 = \{ (z_1,...,z_m) \in \Delta^m \ | \ z_i = a_i, i=2,..,k\}$. It follows that $\mathscr{P}(A,b) \cap H = \mathscr{P}(A,b)^{a_1}\cap H_2$, and in the later we can use the inductive assumption (since $(m-1) - (k-1) = m-k$) to show that if $x_1,...,x_{m-k}$ are affinely independent points in $P_i^{a_1} \cap H_2 = P_i \cap H$, and $y \in Conv(x_1,...,x_{m-k})$ belongs to $P_j$, then $P_i^{a_1}$ and $P_j^{a_1}$ coincide in $H_2$, which is the same as saying $P_i$ and $P_j$ coincide in $H$ as desired.    
\end{proof}

\subsection{Adversarial CD-GBS}

Suppose that $\mathscr{P}$ is an UEPP. Since it is also a proper $(m,n)$-polytope partition, it inherits all the properties from before. Along with Lemma \ref{lemma:structure-uepp} we have the necessary tools to show that Algorithm \ref{alg:GBS-adversarial} is a query efficient way of computing $\epsilon$-close labellings of $\mathscr{P}$ with an adversarial query oracle. In the specification of CD-GBS, we use identical terms and notation from Algorithm \ref{alg:GBS}. 

\begin{algorithm}[h] % enter the algorithm environment
\caption{Adversarial CD-GBS$(m,n,\epsilon, Q_A)$} %ROUGH SKETCH} % give the algorithm a caption
\label{alg:GBS-adversarial} % and a label for \ref{} commands later in the document
\begin{algorithmic} % enter the algorithmic environment
	\INPUT $m \geq 0, \ n,\epsilon > 0$, query access to oracle $Q_A: \Delta^m \rightarrow [n]$.
	\REQUIRE Recursive calls to CD-GBS$ \left( m-1,n,\frac{\epsilon^2}{85(1-x) n m^{5/2}}, Q_A \circ f_x^{-1} \right)$.
	\OUTPUT $\epsilon$-close labelling of $\mathscr{P}$.
	%\INITIALIZATION
	\IF{$m=0$}
	    \STATE Query $Q_A(0)$
	\ELSE
    \STATE $\hat{\mathscr{P}}^0 \leftarrow f_0^{-1} \left( \text{CD-GBS} \left( m-1,n,\frac{\epsilon^2}{85 n m^{5/2}}, Q_A \circ f_0^{-1} \right) \right)$
	\STATE$\hat{\mathscr{P}}^1 \leftarrow Q(\vec{e}_1)$.
	%\STATE
	
	\FOR{$k = 1$ to $\lceil \log(2/\epsilon) \rceil$}
	    \FOR{$x \in \mathscr{D}^k$}
	        \IF{$I_x^k$ is uncovered}
	            \STATE $t \leftarrow midpoint(I_x)$
	            \STATE $\hat{\mathscr{P}}^t \leftarrow f_t^{-1} \left( \text{CD-GBS} \left( m-1,n,\frac{\epsilon^2}{85(1-t) n m^{5/2}}, Q_A \circ f_t^{-1} \right) \right)$
	        \ENDIF
	        \STATE Recompute $\hat{\mathscr{P}}$ by taking convex hulls of labels
	        \WHILE{$\exists i,j \in [n], \ z \in \Delta^m$ such that $dim(\hat{P}_i) = m$ and $z \in int(\hat{P}_i)$}
		        \STATE Merge label $i$ with label $j$
		        \STATE Recompute $\hat{\mathscr{P}}$ by taking convex hulls of labels
		    \ENDWHILE
		    \IF{$\hat{\mathscr{P}}$ is an $\epsilon$-close labelling}
		        \STATE Break
		    \ENDIF
	    \ENDFOR
	\ENDFOR
    \ENDIF
    \iffalse
	\STATE Define $\hat{Q}: \Delta^m \rightarrow [n] \cup \{\bot\}$ as follows:
	\FOR{$i \in [n]$}	
		\STATE if $ x \in \hat{P}_i$,  $\hat{Q}(x) \leftarrow i$, else $\hat{Q}(x) = \bot$
	\ENDFOR
	\fi
	\RETURN $\hat{\mathscr{P}}$    
\end{algorithmic}
\end{algorithm}

\begin{theorem}\label{thm:adv-CDGBS}
If CD-GBS is given access to an adversarial query oracle $Q_A$ of an $(m,n)$-polytope partition based on a UEPP, it computes an $\epsilon$-close labelling of $\mathscr{P}$ using at most\newline
$\left( \prod_{i=1}^m \left( \binom{n+i}{i} + 2n \right) \right) 2^{2m^2}\log^m \left( \frac{170 n m^{5/2}}{\epsilon}  \right)$ membership queries. For constant $m$ this constitutes $O(n^{m^2}\log^m \left( \frac{n}{\epsilon}  \right) ) = poly(n,\log \left( \frac{1}{\epsilon} \right))$ queries.
\end{theorem}

\begin{proof}
As in the proof of correctness of CD-GBS, we begin by noting that when $m=1$ the algorithm runs identical to binary search. We thus focus on the case where $m > 1$.   

The key observation of the proof of correctness is the following: At any given $k$ in the first for loop there are at most $2|C^\alpha_{\mathscr{P}}|$ values of $x$ such that $I^k_x$ is uncovered. Let us consider the empirical polytope $\hat{\mathscr{P}}$ that has been constructed at the time of the execution of the $k$-th loop. Due to the fact that we have merged any labels from the execution of the loop at value $k-1$, it follows that for all $i,j$, $\hat{P}_i \cap \hat{P}_j$ is not of full affine dimension. In turn this means that there exists a hyperplane $H_{i,j}$ that separates the interiors of $\hat{P}_i$ and $\hat{P}_j$. Furthermore, denote $H_{i,j}^+$ as the halfspace defined by $H_{i,j}$ in which $\hat{P}_i$ is contained. This means that in turn we can define $\bar{P}_i = \cap_{j} H_{i,j}^+$ so that $\hat{P}_i \subset \bar{P}_i$. Furthermore, it is straightforward to see that we can define $\bar{\mathscr{P}} = \{ \bar{P}_i\}$ as a valid $(m,n)$-polytope partition that is consistent with $\hat{\mathscr{P}}$. Since $\bar{\mathscr{P}}$ is consistent with our current observations from $Q_A$, we can actually simulate CD-GBS on $\bar{P}_i$ for the first $k-1$ iterations of the algorithm (ordering polytopes accordingly to simulate a lexicographic query oracle). The empirical polytope returned when doing so will in fact be $\hat{\mathscr{P}}$, and thus we can apply corollary \ref{cor:width-bound} to tell us that the number of uncovered $I^k_x$ is in fact bounded by $2|C^\alpha_{\mathscr{P}}|$. 

Returning to our proof of correctness of the algorithm. It is not hard to see that upon termination it is correct if we assume that calling CD-GBS as a subroutine works correctly as per the inductive assumption and crucially the fact that from Lemma \ref{lemma:structure-uepp} each cross-section of $\mathscr{P}$ is non-degenerate. Therefore we focus on the query cost of the algorithm. At each $k =1$ to $\lceil \log(2/\epsilon) \rceil$, from our previous result there can only be at most $2|C^\alpha_{\mathscr{P}}|$ uncovered $I^k_x$, which are precisely the $I^k_x$ that result in queries. By simple multiplication we thus get that the number of cross-section queries is at most $2 \lceil \log(2/\epsilon) \rceil |C_{\mathscr{P}}^\alpha|$, hence we get the following recursion for bounding the query cost of adversarial CD-GBS:
$$
T(m,n,\epsilon) \leq 2\left(\binom{n+m}{m} + 2n \right) \log \left( \frac{2}{\epsilon} \right) T \left(m-1, n, \frac{\epsilon^2}{85 n m^{5/2}} \right)
$$
with base case $T(1,n,\epsilon) \leq n \log \left( \frac{2}{\epsilon} \right)$. If we unpack the recursion in the same way as Theorem \ref{thm:genbinsrch-correct}, we get the desired result.
\end{proof}

\subsection{Adversarial CR-GBS}

In this section we formalize an adversarial variant of CR-GBS. We note that most of the notation is identical to lexicographic CR-GBS. 

\begin{algorithm} % enter the algorithm environment
\caption{CR-GBS$(m,n,\epsilon, \mathscr{P})$} % give the algorithm a caption
\label{alg:FGBS-adversarial} % and a label for \ref{} commands later in the document
\begin{algorithmic} % enter the algorithmic environment
	\INPUT $m,n,\epsilon > 0$, query access to $Q_A$ for $(m,n)$-polytope partition $\mathscr{P}$.
	%\REQUIRE $\mathcal{A}(-,\epsilon)$ as a subroutine for computing $\epsilon$-close labellings of $(k,n)$-polytope partitions.
	\OUTPUT $\epsilon$-close labelling of $\mathscr{P}$.
		\STATE $k \leftarrow \binom{n} {2}$
	\FOR{$F \in \partial( \Delta^m)^k$}
        \STATE $\hat{\mathscr{P}}^F \leftarrow \phi_F^{-1} \left( \text{CD-GBS}\left(k, n, \frac{3\epsilon}{100n^2\sqrt{k+1}(m+1)^{5/2}}, Q \circ \phi^{-1}_F \right) \right)$.
	\ENDFOR
	\STATE $\hat{\mathscr{P}} \leftarrow Conv_F(\hat{\mathscr{P}}_F)$
	\WHILE{$\exists i,j \in [n], \ z \in \Delta^m$ such that $dim(\hat{P}_i) = m$ and $z \in int(\hat{P}_i)$}
		    \STATE Merge label $i$ with label $j$
		    \STATE Recompute convex hulls of labels
	\ENDWHILE
	\iffalse
	\STATE Define $\hat{Q}: \Delta^m \rightarrow [n] \cup \{\bot\}$ as follows:
	\FOR{$i \in [n]$}	
		\STATE if $ x \in \hat{P}_i$,  $\hat{Q}(x) \leftarrow i$, else $\hat{Q}(x) = \bot$
	\ENDFOR
	\fi
	\RETURN $\hat{Q}$  
\end{algorithmic}
\end{algorithm}

\begin{theorem}\label{thm:full-gbs-guarantee-adversarial}
Let $\mathscr{P}$ be an $(m,n)$-polytope partition where $n$ is constant. Furthermore, let $k = \binom{n}{2}$. CR-GBS computes an $\epsilon$-close labelling of $\mathscr{P}$ and uses $O \left( m^k \log^k \left( \frac{m}{\epsilon} \right) \right) = poly(m,\log \left( \frac{1}{\epsilon} \right) )$ queries.
\end{theorem}

\begin{proof}
As in the proof of Theorem \ref{thm:adv-CDGBS}, we know that there exists a polytope partition $\bar{\mathscr{P}}$ that is consistent with $\hat{\mathscr{P}}$. Once again, we notice that this invocation of adversarial CR-GBS is identical to running lexicographic CR-GBS on $\bar{\mathscr{P}}$, which would in turn return $\hat{\mathscr{P}}$ an $\epsilon$-close labelling of $\bar{\mathscr{P}}$. However, it is straightforward to see from the definition of $\epsilon$-close labellings that this also makes $\hat{\mathscr{P}}$ an $\epsilon$-close labelling of $\mathscr{P}$ as desired. Finally, the query usage of adversarial CR-GBS is identical to lexicographic CR-GBS, hence the rest of the theorem follows.
\end{proof}

\section{Games and Best Responses}\label{sec:game-theory-background}

Now that we have established query-efficient algorithms for learning $\epsilon$-close labellings of polytope partitions, we turn our attention to game theory to prove the connection between learning these labellings and computing approximate well-supported Nash equilibria. 

Suppose that $G = (A,B)$ is an $m \times n$ bi-matrix game where  $A,B \in [0,1]^{m \times n}$ are the row player and column player payoff matrices respectively with payoffs normalised to $[0,1]$. We wish to identify an $\epsilon$-well-supported Nash equilibrium ($\epsilon$-WSNE) using only limited information on $G$. The set of row player pure strategies is $[m] = \{1,\ldots,m\}$ and similarly that of the column player pure strategies is $[n] = \{1,\ldots,n\}$. Furthermore, the set of all row player mixed strategies can be associated with the axis-aligned $(m-1)$-simplex: $\Delta^{m-1}= \{ \vec{x} \in \mathbb{R}^{m-1} | \sum_{i=1}^{m-1} x_i \leq 1 \text{ and } x_i \geq 0\}$. Similarly, column player mixed strategies are identified with $\Delta^{n-1}$. 

\begin{definition}[Utility Functions]
Suppose that $u \in \Delta^{m-1}$, and $v\in \Delta^{n-1}$ are row and column player mixed strategies. Let $u' = (1 - \sum u_i, u_1,...,u_{n-1})$ and $v' = (1 - \sum v_i, v_1,...,v_{n-1})$. Then for strategy profile $(u,v)$, row player utility is $U_r(u,v) = u'^TAv'$ and column player utility is $U_c(u,v) = u'^TBv'$.
\end{definition}

It will also be useful to have shorthand for the following functions: $U_r^i(y) = U_r(e_i,y)$ as the row player utility for playing pure strategy $i$, and $E_R(y) = \max_{i \in [m]} U_r^i(y)$ as the maximal utility the row player can achieve against mixed strategies. In an identical fashion we can define $U_c^j$ and $E_C$ as the column player utility in playing strategy $j$ and the maximal column player utility. With this notation in hand, we can define the best response oracles algorithms will have access to when computing approximate Nash equilibria.

\begin{definition}[Best Response Query Oracles]
Any bimatrix game has the following best response query oracles:

\begin{itemize}
    \item Strong query oracles: for the column player, $BR^C_s(u) = \{ j \in [n] \ | \ U^j_c(u) = E_C(u) \}$ and for the row player, $BR^R_s(v) = \{ i \in [m] \ | \ U^i_r(v) = E_R(v) \}$
    \item Lexicographic query oracles: for the column player, $BR^C_\ell(u) = \min_{j \in [n]} j \in BR^C_s(u)$ and for the row player, $BR^R_\ell(v) = \min_{ i \in [m]} i \in B^R_s(v)$
    \item Adversarial query oracles: for the column player, any function $BR^C_A$ such that $BR^C_A(u) \in BR^C_s$ and for the row player, any function such that $BR^R_A(v) \in BR^R_s(v)$
\end{itemize}
\end{definition}

For a given mixed strategy, $u \in \Delta^{m-1}$, we say the support of $u$ is the set of pure strategies that are played in $u$ with non-zero probability. It will be useful to formulate this as a function in order to define Nash equilibria combinatorially in the following section.

\begin{definition}[Support Functions] 
Let $S^R : \Delta^{m-1} \rightarrow \mathcal{P}([m])$ be the function which returns the support of a row player mixed strategy. Similarly let $S^C : \Delta^n \rightarrow \mathcal(\mathcal{P}[n])$
return the support of column player mixed strategies.
\end{definition}

\section{Nash Equilibria and Lipschitz Continuity of Utility Functions}\label{sec:nash-lipschitz}

\begin{definition}[Nash Equilibrium]\label{def:nash}
Suppose that $u$ and $v$ are row and column player strategies respectively. We say that the pair $(u,v)$ is a {\em Nash Equilibrium (NE)} if for all $u' \in \Delta^{m-1}$ and $v' \in \Delta^{n-1}$: $U_r(u,v) \geq U_r(u',v)$ and $U_c(u,v) \geq U_c(u,v')$. 
\end{definition}

Though the definition of a Nash equilibrium involves utility values of both players at their mixed strategy profiles, there is an equivalent combinatorial formulation of the above definition:

\begin{proposition}\label{prop:comb-NE}
$(u,v)$ is a NE if and only if $S^R(u) \subseteq BR^R_s(v)$ and $S^C(v) \subseteq BR^C_s(u)$. In other words $u$ is supported by best responses to $v$ and vice versa.
\end{proposition}

When using best response queries only, one does not have access to utility values (as emphasised in the first definition), however this second equivalent definition of Nash equilibria can be verified by using best response oracles and support functions alone.\footnote{In general one is unable to recover utility values from best responses, even up to affine transformations.} 

We also note that for utility queries, the complexity of an exact NE is finite: we can exhaustively query the game. On the other hand, Corollary \ref{cor:exact-nash-inft} shows that this is not the case for best response queries. As a consequence, we relax the notion of a NE when using best response queries. the relaxation of NE which we study is that of approximate well-supported equilibria. Before proceeding with the formal definition, we say that a row player mixed strategy $u \in \Delta^m$ is an $\epsilon$ best response against a column player mixed strategy $v \in \Delta^n$ if $U_r(u,v) \geq U_r(u',v) - \epsilon$ for all $u' \in \Delta^m$. An identical notion holds for when a column player mixed strategy $v \in \Delta^n$ is an $\epsilon$ best response against a row player mixed strategy $u \in \Delta^m$. Intuitively, an $\epsilon$ best response is a mixed strategy where a player has only an $\epsilon$ incentive to deviate.

\begin{definition}[$\epsilon$-Well-Supported Nash Equilibrium]\label{def:eps-WSNE}
Suppose that $u$ and $v$ are row and column player strategies respectively. We say that the pair $(u,v)$ is an $\epsilon$-well-supported Nash equilibrium ($\epsilon$-WSNE) if and only if $u$ is supported by $\epsilon$-best responses to $v$ and vice versa.\footnote{Note that the conditions for an $\epsilon$-WSNE imply that no player has more than an $\epsilon$ incentive to deviate from the approximate equilibrium.}
\end{definition}

\begin{theorem}\label{thm:WSNE-LB}
The query complexity of computing an $\epsilon$-WSNE with best response queries is $\Omega( \log \left(\frac{1}{\epsilon} \right) )$, even when given access to strong query oracles.
\end{theorem}

\begin{proof}
Suppose that $x,y \in (0,1)$ are arbitrary, and let us consider a two-player binary action game, $G_{x,y}$, with the following row and column player payoff matrices:

\[A_x =  \left( \begin{array}{rr}
x & x \\
0 & 1 
\end{array} \right)
\ \ 
B_y =  \left( \begin{array}{rr}
0 & y \\
1 & y 
\end{array} \right)
\]
Since $G_{x,y}$ is a binary action game, we can express any mixed strategy profile of both players by a tuple $(p_r,p_c) \in [0,1]^2$. $p_r$ represents the probability the row player plays the second row and $p_c$ represents the probability that the column player plays the second column. 

It is clear that this game has no pure equilibria, but its unique NE is the mixture: $(y,x) \in [0,1]^2$. Upon close inspection, one can also see that the set of $\epsilon$-WSNE of $G_{x,y}$ lie in $B_\epsilon(y) \times B_\epsilon(x) \cap [0,1]^2$, where $B_\epsilon(z)$ the set of points at a distance $\epsilon$ from $z$ in $\mathbb{R}$. Suppose that an algorithm, $\mathcal{A}$, is given a game $G_{x,y}$ from the family above and access to the game's strong best response query oracles. In order for $\mathcal{A}$ to compute an $\epsilon$-WSNE, the previous observation tells us that the it has to at least find a point $z \in B_\epsilon(x)$ by querying the row player's best response oracle. From the structure of $A_x$ however, we know that for $p_c \in [0,x]$ the first row is a best response for the row player, and for $p_c \in [x,1]$, the second row is a best response for the row player. This problem formulation however is identical to binary search, hence finding a $z \in B_\epsilon(x)$ takes at least $\Omega( \log \left(\frac{1}{\epsilon} \right) )$ queries.

\iffalse
If it had a finite query complexity, it would make say $k$ queries $q_1,...,q_k$ to $BR^R_s$. Since the row player plays strategies from $\Delta^1$, these can be parametrized as values from $[0,1]$ which is the probability the row player plays the second row. Now suppose that $q_i \in [0,1]$ is the largest query made by the algorithm such that $BR^R_s(q) = c_2$, and furthermore suppose that $q_j \in [0,1]$ is the smallest query made by the algorithm such that $BR^R_s(q) = c_1$. From convexity of best response sets it is clear that $q_i \neq q_j$ and $q_i \leq y \leq q_j$. If the algorithm is to correctly return the exact NE of the game, it must return a $y \in [q_i, q_j]$, but an adversary can choose a game with true equilibrium $y' \neq y$ such that $y' \in [q_i,q_j]$.  
\fi
\end{proof}

\begin{corollary}\label{cor:exact-nash-inft}
The query complexity of computing a NE with best response queries is infinite, even when given access to strong query oracles.
\end{corollary}

\subsection{Algebraic Properties of Utility Functions}

Definition \ref{def:eps-WSNE} mentions approximate best responses, yet we only have access to the best response oracle in our model. In order to resolve this, we delve into the algebraic properties of utility functions of both the column and row player.

\begin{lemma}\label{lemma:utilities-lipschitz}
If the domains of $U^i_r$ and $U^j_c$ are endowed with the $\ell_2$ norm, then the functions are $\lambda_R$ and $\lambda_C$ Lipschitz continuous respectively, for some $0\leq \lambda_R \leq \sqrt{n-1}$ and $0 
\leq \lambda_C \leq \sqrt{m-1}$. If the domains are endowed with the $\ell_1$ norm, then both functions are $1$-Lipschitz continuous.
\end{lemma}

\begin{proof}
We focus on $U^i_r$, the case for the column player is identical. Let $c = [A^T]_i = (a_1,...,a_n)$ be the $i$-th row vector of the row player's payoff matrix, and suppose that $v = (v_1,...,v_{n-1}) \in \Delta^{n-1}$ is a column player mixed strategy, where $v_0 = 1 - \sum_{i=1}^{n-1} v_i$ is implicit. Let $z = (z_i)_{i=1}^{n-1}$ with $z_i = (a_i - a_0)$ for $i= 1,...,n-1$. Then it is clear that $U^i_r(v) = a_0 + \sum_{i=1}^{n-1}z_i \cdot v_i$. This function is linear, and trivially $\|z\|_2$-Lipschitz continuous. Since the game is normalised, $\|z\|_2 \leq \|\vec{1}\|_2 = \sqrt{n-1}$.

As for the second part of the claim, the domain of $U^i_r$ can be equivalently represented as $\Lambda^{n} = \{x \in \mathbb{R}^{n} \ | \ \|x\|_1 = 1$, $x_i \geq 0\}$ by using the invertible linear map $\phi_{n-1}: \Delta^{n-1} \rightarrow \Lambda^{n}$ given by $\phi_{n-1}(x_1,...,x_{n-1}) = \left( x_1,...,x_{n-1}, \left( 1-\sum_{i=1}^{n-1}{x_i} \right) \right)$. This space can be endowed with total variation distance as a metric, which for two distributions, $x,y \in \Lambda^n$ is defined as $TV(x,y) = \max_{s \subseteq n} | \mathbb{P}_x(S) - \mathbb{P}_y(S)| = \frac{1}{2} \|x - y\|_1$. Since utilities are bounded to be in the interval $[0,1]$, it follows that $U^i_r$ is 1-Lipschitz as a function with domain $\Lambda^n$ in the total variation metric.

Now suppose that $x,y \in \Delta^{n-1}$. We wish to show that $TV(\phi_{n-1}(x), \phi_{n-1}(y)) \leq \|x - y\|_1$. To see this, let $x_n = 1 - \sum_{i=1}^{n-1} x_i$ and $y_n = \sum_{i=1}^{n-1} y_i$. Then $\|\phi_{n-1}(x) - \phi_{n-1}(y)\|_1  = \|x - y\|_1 + |x_n - y_n|$. From this we see that $|x_n - y_n| = |\sum_{i=1}^{n-1}(y_i - x_i)| \leq \|x-y\|_1$, which in turn implies $\|\phi_{n-1}(x) - \phi_{n-1}(y)\|_1 \leq 2\|x- y\|_1$. Dividing the expresion by 2 and applying the fact that $TV(x,y) = 
\frac{1}{2}\|x-y\|_1$ proves our desired inequality. 1-Lipschitz continuity in the $\ell_1$ norm for domain $\Delta^{n-1}$ follows immediately.
\end{proof}

\begin{corollary}
Since $E_C$ and $E_R$ are defined as a pointwise maximum, it follows that they are also Lipschitz continuous with constant $\lambda_C \leq \sqrt{m-1}$ and $\lambda_R \leq \sqrt{n-1}$ in the $\ell_2$ norm. 
\end{corollary}

With bounded Lipschitz continuity, we have guarantees on how much utilities can deviate between ``close'' mixed strategy profiles. This has interesting implications even for the best response query oracle, for this means that if $u$ and $u'$ are close in the $\ell_2$ norm with $c_i \in BR^C_s(u)$, $ c_j \in BR^C_s(u')$ and $c_i \neq c_j$, then we can say that $c_i$ and $c_j$ are both approximate best responses in the vicinity of $u$ and $u'$. We formalise this as follows.

\begin{lemma} \label{lemma:close-epsBR}
Fix $\epsilon >0$ and let $\delta_C = \frac{\epsilon}{2\sqrt{m-1}}$. Suppose that $u \in \Delta^{m-1}$ is a row player mixed strategy with $c_j \in BR^C_s(u)$. For any $u'$ such that $\|u - u'\|_2 \leq \delta_C$, if $c_i \in BR^C_s(u')$, then $|U^i_c(u) - U^j_c(u)| \leq \epsilon$. In other words, $c_i$ is an $\epsilon$-best response to $u$. Similarly, let $\delta_R = \frac{\epsilon}{2\sqrt{n-1}}$. Suppose that $v \in \Delta^{n-1}$ is a column player mixed strategy with $r_j \in BR^R_s(u)$. For any $v'$ such that $\|v - v'\|_2 \leq \delta_R$, if $r_i \in BR^R_s(v')$, then $|U^i_r(v) - U^j_r(v)| \leq \epsilon$. In other words, $r_i$ is an $\epsilon$-best response to $v$.
\end{lemma}

\begin{proof}
Suppose that $u'$ is such that $\|u - u'\| \leq \delta_C$ and $c_i \in BR^C_s(u')$. By definition, $E_c(u') = U^i_c(u') \geq U^j_c(u')$ and by Lemma \ref{lemma:utilities-lipschitz}, $|U_j^c(u) - U_j^c(u')| \leq \lambda_C \|u - u'\|_2$, and $\|U_i^c(u) - U_i^c(u')| \leq \lambda_C \|u - u'\|_2$. With these expressions we obtain the following inequalities:
$$
|E_c(u) - U_c^i(u)| \leq 2\lambda_C\|u - u'\|_2 \leq 2\lambda_C \frac{\epsilon}{2\lambda_C} = \epsilon
$$
The proof of the second half of the lemma is identical.
\end{proof}

The previous Lemma establishes the important idea that we can obtain some information regarding approximate best responses using only the best response oracle and ``nearby'' queries. With some thought one can see that this in general does not reveal {\em all} approximate best response information. For example, if a strategy were strictly dominated, a best response oracle would never see it, and hence never be able to tell if it was an approximate best response. 

\section{Nash's Theorem with Discrete Approximations}\label{sec:discrete-nash} 

We are now in a position to prove the intimate connection between computing $\epsilon$-close labellings of upper envelope polytope partitions and computing $\epsilon$-WSNE for bimatrix games using best response queries.

\begin{definition}[Best Response Sets]\label{BR-sets}
Let $G = (A,B)$ be a bimatrix game. We define column best response sets as the collection of $C_i = \{ x \in \Delta^{m-1} \ | \ BR^C_s(x) = c_i \}$. Similarly we define row player best response sets as the collection of $ R_j = \{ y \in \Delta^{n-1} \ | \ BR^R_s(y) = r_j \}$. We denote the collections by $\mathscr{C} =\{C_i\}_{i=1}^n$ and $\mathscr{R} = \{R_j\}_{j=1}^m$.
\end{definition}

Since utilities are affine functions, it is immediately clear that $\mathscr{C}$ and $\mathscr{R}$ are upper envelope polytope partitions. Now the best response oracles play the same role as membership oracles, $Q$, from before. Since adversarial oracles are the weakest of the three membership oracles (in the sense that they are a valid lexicographic oracle and they can be simulated with access to a strong oracle), we focus on using adversarial best response oracles. Furthermore, with our language of empirical labellings we can now define a key object used in the computation of approximate equilibria. Before doing so, we clarify some notation: $d(x,S)$ denotes the infimum distance of a point, $x$ to a set $S$.

\begin{definition}[Voronoi Best Response Sets]\label{def:VorBR}
Suppose that $\hat{\mathscr{C}} = \{\hat{C}_i\}$ and $\hat{\mathscr{R}} = \{\hat{R}_j\}$ are empirical labellings of $\mathscr{C}$ and $\mathscr{R}$ as in Definition \ref{def:implicit-labelling}. The {\em Voronoi Best Response Sets} of the row and column player are $VR_j = \{ y \in \Delta^{n-1} \ | \ \argmin_j d(y,\hat{R}_j) = r_j \}$ and $VC_i = \{ x \in \Delta^{m-1} \ | \ \argmin_i d(x,\hat{C}_i) = c_i \}$, defined for any $j \in [m]$ and $i \in [n]$. Furthermore, we let $V^R(v) = \{i \ | \ VR_i \ni v\}$ and $V^C(u) = \{j \ | VC_j \ni u\}$ be the row and column player {\em Voronoi Best Responses}.
\end{definition}

\begin{lemma}\label{lemma:voronoi-closed}
Voronoi best response sets partition $\Delta^{m-1}$ and $\Delta^{n-1}$ into closed connected regions with non-empty interior and piecewise linear boundaries.
\end{lemma}

\begin{proof}
Without loss of generality, let us focus on a given column player Voronoi best response set; i.e. some $VC_i \subset \Delta^{m-1}$ that is non-empty and arises from the empirical labelling $\hat{\mathscr{C}} = \{\hat{C}_i\}$ of column-player best responses. First we note that $\hat{C}_i \subset VC_i$ by definition, and the former is a convex, closed, and connected polytope of $\Delta^{m-1}$. Therefore it remains to show that if we pick an arbitrary $x \in VC_i \setminus \hat{C}_i$ it is connected to $C_i$. To do so, suppose that $p \subset \Delta^{m-1}$ is the unique shortest path from $x$ to $C_i$. It is clear that all points along $p$ must also lie in $VC_i$, therefore the claim holds. 

As for closedness, note that if $\{x_n\}$ is a sequence in $VC_i$ that converges to some $x \in \Delta^{m-1}$ then $x$ must also be in $VC_i$. This follows from the continuity of Euclidian distance for $\Delta^{m-1} \subset \mathbb{R}^{m}$. Now suppose that $x$ is a limit point of $VC_i$, then we can construct a sequence as above, and thus $x$ must also be in $VC_i$, rendering the set closed.

Finally, the piecewise linear boundary arises from the fact that $\hat{C}_i$ is itself a closed convex polytope which has a piecewise linear boundary. Decision boundaries between different $\hat{C}_i$ and $\hat{C}_j$ are composed of nearest neighbour decision boundaries between piecewise linear boundaries, which in turn results in piecewise linear decision boundaries between $VC_i$ and $VC_j$.    
\end{proof}

Although the previous lemma proves that Voronoi best response sets partition $\Delta^{m-1}$ and $\Delta^{n-1}$ into closed connected regions with non-empty interior and piecewise linear boundaries, they need not be convex. This ends up not being an issue for our subsequent results. The reason we deal with these objects however is due to the following Lemma. We recall that $\lambda_R \leq \sqrt{m-1}$ and $\lambda_C \leq \sqrt{n-1}$ are the relevant Lipschitz continuity constants for row player and column player expected utility functions. The following is a straightforward consequence of Lemma \ref{lemma:close-epsBR}.
 
\begin{lemma}\label{lemma:Vor-to-epsBR}
Suppose that $\hat{\mathscr{C}}$ is a $\frac{\epsilon}{2\sqrt{m-1}}$-close labelling and $\hat{\mathscr{R}}$ is a $\frac{\epsilon}{2\sqrt{n-1}}$-close labelling. Then Voronoi best responses are $\epsilon$ best-responses in $G$  
\end{lemma}

We recall that the combinatorial formulation of Nash's theorem implies that with full information of best response sets in all of $\Delta^m$ and $\Delta^n$, one is able to compute and verify an exact Nash equilibrium. Best responses only partially recover this information in convex patches of $\Delta^m$ and $\Delta^n$. Furthermore, it is not clear how a game $G'$ with best responses sets consistent with empirical best response sets of $G$ can be used to compute an approximate equilibrium of $G$. Voronoi best response sets however allow us to take the partial information provided by empirical best response sets and extend it to approximate best response information {\em across the entire domains $\Delta^m$ and $\Delta^n$} (Voronoi best response sets cover $\Delta^m$ and $\Delta^n$ after all). This hints at the fact that Voronoi best response sets hold enough information to compute $\epsilon$-WSNE. In fact we can prove this in the same way as Nash's theorem: via Kakutani's fixed point theorem. In order to do so, we define a Voronoi best response correspondence (which as we have shown before is an approximate best response correspondence), and show that it satisfies the properties of Kakutani's fixed point theorem. The guaranteed fixed point of this correspondence will in turn be an $\epsilon$-WSNE. 
 
\begin{definition}[Voronoi Approximate Best Response Correspondence]\label{def:VorBR-correspondance}
For a given mixed strategy profile of both the row and column player, $(u,v) \in \Delta^{m-1} \times \Delta^{n-1}$, we define $B^*(u,v)$ to be the set of all possible mixtures over Voronoi best response profiles both players may have to the other player's strategy. $B^*: \Delta^{m-1} \times \Delta^{n-1} \rightarrow \mathcal{P}(\Delta^{m-1} \times \Delta^{n-1})$ is defined as follows:

$$
B^*(u,v) = \left( conv(V^R(v)), conv(V^C(u)) \right) \subseteq \Delta^{m-1} \times \Delta^{n-1}.
$$ 
\end{definition} 
 
\begin{theorem}[Kakutani's Fixed Point Theorem \cite{kakutani1941}]\label{thm:kakutani}
Let $A$ be a non-empty subset of a finite-dimensional Euclidian space and $f:A \rightarrow \mathcal{P}(A)$ be a set-valued function satisfying the following conditions:
\begin{itemize}

\item $A$ is a compact and convex set.
\item $f(x)$ is non-empty for all $x \in A$.
\item $f(x)$ is a convex-valued correspondence: for all $x \in A$, $f(x)$ is a convex set.
\item $f(x)$ has a closed graph: that is, if $\{x_n,y_n\} \rightarrow \{x,y\}$ with $y_n \in f(x_n)$ for all $n$, then $y \in f(x)$. 

\end{itemize}
Then $f$ has a fixed point, that is there exists some $x \in A$ such that $x \in f(x)$. 
\end{theorem}

\begin{theorem}\label{thm:discrete-Nash}
$B^*$ satisfies all the conditions of Kakutani's fixed point Theorem, and hence there exists a strategy profile $(u^*,v^*)$ such that $(u^*,v^*) 
\in B^*(u^*,v^*)$. In particular, if the Voronoi best responses for $B^*$ arise from $\hat{\mathscr{C}}$, a $\frac{\epsilon}{2\sqrt{m-1}}$-close labelling and $\hat{\mathscr{R}}$, a $\frac{\epsilon}{2\sqrt{n-1}}$-close labelling, then this in turn implies that $(u^*,v^*)$ is an $\epsilon$-WSNE of $G$. 
\end{theorem}

\begin{proof}
We need to prove the following conditions for Kakutani's fixed point Theorem:
\begin{itemize}

\item $B^*$ has a compact and convex domain.
\item $B^*(u,v)$ is non-empty and convex for all $(u,v) \in \Delta^{m-1} \times \Delta^{n-1}$.
\item (Graph Closedness) Suppose that $\{\sigma_n\}$ and $\{\sigma_n'\}$ are sequences in $\Delta^{m-1} \times \Delta^{n-1}$ that converge to $\sigma$ and $\sigma'$ respectively. Furthermore, suppose that $\sigma'_n \in B^*(\sigma_n)$ for all $n$. Then $\sigma' \in B^*(\sigma)$.    

\end{itemize} 
For the first item, the domain of $B^*$ is $\Delta^{m-1} \times \Delta^{n-1}$ which clearly satisfies the desired condition.

As for the second and third item, from the definition of $B^*$ the image of any $(u,v)$ consists of convex combinations of Voronoi best responses, which are defined for all $(u,v)$ (thus satisfying non-emptyness), and since they are convex combinations, they are convex subsets of $\Delta^{m-1} \times \Delta^{n-1}$.

Finally for the fourth item, let us consider such a sequence where $\sigma_n = (u_n,v_n)$, $\sigma = (u,v)$, $\sigma'_n = (u_n',v_n')$, and $\sigma' = (u',v')$. To show the claim, it suffices to consider the sequences $\{u_n\}$ and $\{v_n'\}$ with respective limits $u$ and $v'$, and show that $v' \in conv(V^C(u))$ 

To show this however, it suffices to use the fact that Voronoi best response sets are closed. Suppose that $u$ has a certain set $S \subset [n]$ of Voronoi best responses. Then there exists a constant $\mu > 0$ such that $B_\mu(u) \cap VC_i \neq \emptyset$ if and only if $i \in S$; namely the $\mu$ neighbourhood around $u$ only intersects Voronoi best response sets from $u$'s Voronoi best responses.

To explicitly construct such a $\mu$, let us consider $D_i = d(u,\hat{C}_i)$ to be the distance between $u$ and the empirical best response set $\hat{C}_i$. This means that $S = \argmin_i D_i$, so let us define $D = \min_i D_i$ and $\mu = \frac{\min_{j \notin S} D_j - D}{3} $ which is positive due to the fact that there are finitely many partial best response sets. Now suppose that $x \in B_\mu(u)$, then for any $j \notin S$ we have $d(u,\hat{C}_j) \leq d(u,x) + d(x,\hat{C}_j)$ by the triangle inequality, which rearranging gives us: $d(x,\hat{C}_j) \geq d(u,\hat{C}_j) - d(u,x) \geq (D + 3\mu) - \mu = D +2\mu$. On the other hand, for any $ i \in S$, $d(x,\hat{C}_i) \leq d(x,u) + d(u,\hat{C}_i) \leq D + \mu$. It thus follows that $x$ can only have Voronoi best responses from $S$. 

Now from the fact that $u_n \rightarrow u$, then for some $N> 0$, if $n >N$ then $u_n \in B_\mu(u)$. This in turn means that $v_n \in conv(S)$ by assumption, which means that $v \in conv(S)$ as well, which is what we wanted to show. To extend this to $\sigma$ and $\sigma'$, it suffices to repeat the previous argument in each component of the correspondence.   
 
Now that the conditions of Kakutani's fixed point Theorem are satisfied, we know of the existence of an $(u^*,v^*)$ such that $(u^*,v^*) \in B^*(u^*,v^*)$. As in the statement of the Theorem, suppose that Voronoi best responses for $B^*$ arise from an $\frac{\epsilon}{2\sqrt{m-1}}$-close labelling of $\Delta^m$ and an $\frac{\epsilon}{2\sqrt{n-1}}$ of $\Delta^n$, then we know that all Voronoi best responses are $\epsilon$ best responses for both players. The conditions of the fixed point amount to saying that both players are playing convex combinations of Voronoi best responses, therefore $(u^*,v^*)$ is an $\epsilon$-WSNE. 
\end{proof} 

With Theorem \ref{thm:discrete-Nash} in hand and our algorithms for constructing $\epsilon$-close labellings, we can put everything together and prove our desired results regarding the query complexity of computing an $\epsilon$-WSNE in general bimatrix games.

\begin{theorem}\label{thm:final-nash-BR}
Suppose that $G$ is an $m \times n$ bimatrix game and let $n$ be constant. We can compute an $\epsilon$-WSNE using $O(m^{n^2}\log^{n^2} \left( \frac{m}{\epsilon} \right) ) = poly(m, \log \left( \frac{1}{\epsilon} \right) )$ adversarial best response queries.
\end{theorem}

\begin{proof}
Suppose that $\mathscr{C}$ and $\mathscr{R}$ are the polytope partitions arising from best-response sets in $G$. This means that $\mathscr{C}$ is a $(m-1,n)$-polytope partition and $\mathscr{R}$ is a $(n-1,m)$-polytope partition. Let $\epsilon_C = \frac{\epsilon}{2\sqrt{m-1}}$ and $\epsilon_R = \frac{\epsilon}{2\sqrt{n-1}}$. From Theorem \ref{thm:discrete-Nash}, we know that computing an $\epsilon_C$-close labelling of $\mathscr{C}$ and a $\epsilon_R$-close labelling of $\mathscr{R}$ suffice to compute an $\epsilon$-WSNE of $G$. We use adversarial CR-GBS on $\mathscr{C}$ and adversarial CD-GBS on $\mathscr{R}$.

$n$ is the number of polytopes in the partition $\mathscr{C}$, which is assumed to be constant. Consequently, Theorem \ref{thm:full-gbs-guarantee-adversarial} states that computing an $\epsilon_C$-close labelling of $\mathscr{C}$ using CR-GBS uses $O((m-1)^k\log^k \left( \frac{m-1}{\epsilon_C} \right))$ adversarial queries, where $k = \binom {n} {2}$. Since $k \leq n^2$, we can upper bound the number of queries by $O(m^{n^2}\log^{n^2} \left( \frac{m}{\epsilon} \right) )$.

$n-1$ is the dimension of the ambient simplex in the partition $\mathscr{R}$, which is assumed to be finite. Consequently, Theorem \ref{thm:adv-CDGBS} states that computing an $\epsilon_R$-close labelling of $\mathscr{R}$ using CD-GBS uses $O(m^{(n-1)^2}\log^{(n-1)^2} \left( \frac{1}{\epsilon} \right) )$ queries. We trivially upper bound this quantity by $O(m^{n^2}\log^{n^2} \left( \frac{m}{\epsilon} \right) )$.

Putting everything together, the total query usage is thus $O(m^{n^2}\log^{n^2} \left( \frac{m}{\epsilon} \right) ) = poly(m, \log \left( \frac{1}{\epsilon} \right) )$ as desired.
\end{proof}

\section{A Brief Foray into Multiplayer Games}

In this section we partially extend our results from sections \ref{sec:nash-lipschitz} and \ref{sec:discrete-nash}. In particular, we show that in multiplayer games utility functions are Lipschitz continuous as well, which allows us to uncover approximate best-response information when observing different best responses at ``nearby'' mixed strategy profiles. In addition we generalise our definitions of best response sets and of $\epsilon$-close labellings to obtain a result similar to Theorem \ref{thm:discrete-Nash} where we showed that obtaining a precise enough empirical labelling provides enough information to compute a well-supported approximate Nash equilibrium. 

The main difference in the multiplayer setting however is that best response sets are no longer polytopes nor convex, which means that our algorithms for computing empirical labellings via generalisations of binary search no longer apply. This does not preclude us however from simply querying an $\epsilon$-net, which as we will show will suffice for computing an $\epsilon$-WSNE.

\subsection{Notation for Multiplayer Games}

For simplicity we will focus on games with $n$ players where each player has a strategy set $A_i$ consisting of $|A_i| = k$ pure strategies. It is straightforward to extend our results to more general games where different players have action sets of different cardinalities. 

In general, we let $A = \prod_{i=1}^n A_i$ be the space of all pure strategy profiles of all players. For the $i$-th player, we also denote $A_{-i} = \prod_{j \neq i} A_j$ as the space of all pure strategy profiles of players other than $i$. Since every player has $k$ actions, it is straightforward to see that all $A_{-i}$ are isomorphic, hence without loss of generality we can assume that for all $i$, $A_{-i}$ is canonical representation. For a given pure strategy profile $a \in A$, we may wish to distinguish the pure strategy taken by the $i$-th player, and this is done by writing $a = (a_i, a_{-i})$ with $a_i \in A_i$ and $a_{-i} \in A_{-i}$.

We denote the $i$-th player's mixed strategy space by $\Delta(A)_i$, and we note it is equivalent to $\Delta^{k-1}$. This means that the space of mixed strategy profiles of all players is $\Delta(A) = \prod_{i=1}^n (\Delta^{k-1}) = (\Delta^{k-1})^n$. In addition, for the $i$-th player, we also denote $\Delta(A)_{-i} = \prod_{i=1}^n (\Delta^{k-1}) = (\Delta^{k-1})^{n-1}$ as the space of all mixed strategy profiles of players other than $i$. Once again, since all players have $k$ actions, we can assume without loss of generality that all $\Delta(A)_{-i}$ have the same canonical representation. For a mixed strategy profile $x \in \Delta(A)$, we may wish to distinguish the mixed strategy of the $i$-th player by writing $x = (x_i, x_{-i})$ with $x_i \in \Delta(A)_i$ and $x_{-i} \in \Delta(A)_{-i}$.

\begin{definition}[Multiplayer Utility Functions] \label{def:multi-utility}
For any player $i$ and action $r \in A$, we denote $U_i^r:A_{-i} \rightarrow [0,1]$ as the $i$-th player's utility for playing $r$. If $a = (a_i,a_{-i}) \in A$ is a pure strategy profile, the utility player $i$ receives is $U_i^{a_i}(a_{-i})$ which we denote by $U_i(a)$. 
\end{definition}

Utility functions are defined for pure strategy profiles, but with a slight abuse of notation we extend the domain to include mixed strategy profiles. In particular, if $x = (x_i,x_{-i})$ is a mixed strategy profile, we let $U_i^r(x_{-i}) = \mathbb{E}_{a_{-i} \sim x_{-i}}(U_i^r(a_{-i}))$ and by extension $U_i(x) = \mathbb{E}_{a \sim x} (U_i(a))$. 

As in bimatrix games, for a given player $i$ and $x_{-i} \in \Delta(A)_{-i}$, we let $E_i(x_{-i}) = \max_{r \in A_i} U_i^r(x_{-i})$ be the maximal expected utility player $i$ can obtain against the mixed strategy profile $x_{-i}$ of all other players.  

\begin{definition}[Best Response Query Oracles]\label{def:BR-oracle-multiplayer}
Let $G$ be an $n$-player game where each player has $k$ pure strategies:
\begin{itemize}
    \item There are $n$ strong best response oracles denoted $BR^i_s: \Delta(A)_{-i} \rightarrow \mathcal{P}(A_i)$ for $i=1,...,n$. Each of these is defined by $BR^i_s(x) = \{r \in A_i \ | \ U_i^r(x) = E_i(x)\}$.  
    \item There are $n$ lexicographic best response oracles denoted $BR^i_\ell: \Delta(A)_{-i} \rightarrow A_i$ for $i=1,...,n$. Each of these is defined by $BR^i_\ell(x) = \argmin_{r \in A_i} r \in BR^i_s(x)$ for some consistent order on each $A_i$.  
    \item An adversarial best response oracle collection is a collection of $n$ functions denoted $BR^i_A: \Delta(A)_{-i} \rightarrow A_i$ for $i=1,...,n$. Each of these satisfies $BR^i_A(x) \in BR^i_s(x)$. 
\end{itemize}
\end{definition}

For completeness we have defined all best response oracles, but we focus on adversarial best response oracles. As mentioned before, they are the weakest from the fact that a lexicographic oracle is a valid adversarial oracle and the fact that adversarial oracles can be simulted with strong best response oracles. This implies that our results for adversarial oracles carry over to other oracle models.

For a given mixed strategy, $x \in \Delta(A)_i$, we say the support of $x$ is the set of pure strategies that are played in $x$ with non-zero probability. It will be useful to formulate this as a function in order to define Nash equilibria combinatorially again.

\begin{definition}[Support Functions] 
Let $S^i : \Delta(A)_i \rightarrow \mathcal{P}(A_i)$ be the function which returns the support of the $i$-th player's mixed strategy.
\end{definition}

We are now in a position to define what a Nash equilibrium is in multiplayer games.

\begin{definition}[Nash Equilibrium]\label{def:nash-multiplayer}
We say $x \in \Delta(A)$ is a Nash Equilibrium (NE) if for any player $i$, and $x_i' \in \Delta(A)_i$ it holds that $U_i(x) \geq U_i(x'_i, x_{-i})$

\end{definition}

As before, this definition involves utility values of players at their mixed strategy profiles. Once again, there is an equivalent combinatorial formulation.

\begin{proposition}\label{prop:comb-NE-multiplayer}
$x \in \Delta(A)$ is a NE if and only if for all $i$, when $x = (x_i, x_{-i})$, $S^i(x_i) \subseteq BR^i_s(x_{-i})$.
\end{proposition}

Finally, we define what it means for a mixed strategy profile to be an $\epsilon$-WSNE in the multiplayer setting. Before proceeding, we say that for a given $x_{-i} \in \Delta(A)_{-i}$, strategy $x \in \Delta(A)_i$ is an $\epsilon$ best response if $U_i(x,x_{-i}) \geq U_i(x',x_{-i}) - \epsilon$ for all $x'\in \Delta(A)_i$. Intuitively, an $\epsilon$ best response is a mixed strategy where a player has only an $\epsilon$ incentive to deviate.

\begin{definition}[$\epsilon$-Well-Supported Nash Equilibrium]\label{def:eps-WSNE-multiplayer}
Suppose that $x \in \Delta(A)$ is a mixed strategy profile of all players We say that $x$ is an $\epsilon$-well-supported Nash equilibrium ($\epsilon$-WSNE) if and only if for every player $i$, all pure strategies in $S^i(x_i)$ are $\epsilon$ best responses to $x_{-i}$. 
\end{definition}

\subsection{Lipschitz Continuity of Utility Functions}

As in Section \ref{sec:nash-lipschitz}, we will show that for each player $i$ and each $r \in A_i$, $U_i^r$ is a Lipschitz continuous function. In order to do so, we must regard the domain of $U_i^r$, which is $\Delta(A)_{-i} = (\Delta^{k-1})^{n-1}$, as a subset of Euclidean space endowed with the $\ell_1$ norm.

\begin{lemma} \label{lemma:utility-lipschitz-multiplayer}
For any player $i$ and action $r \in A_i$, $U_i^r$ is 1-Lipschitz when the domain is endowed with the $\ell_1$ norm. 
\end{lemma}

\begin{proof}

Consider an arbitrary player $j$ with $j \neq i$. If we endow $\Delta(A)_j$ with the $\ell_1$ norm, from Lemma \ref{lemma:utilities-lipschitz}, we know that $U_i^r$ as a function of $x_j \in \Delta(A)_j$ (which is a component of $\Delta(A)_{-i}$) is 1-Lipschitz. This is because if all mixed strategies other than those of player $i$ and $j$ are fixed, we obtain a $k \times k$ bimatrix game between player $i$ and $j$. Now let us consider $x,y \in \Delta(A)_{-i}$.

\begin{equation} %\label{eq1}
\begin{split}
|U_i^r(x) - U_i^r(y)| &  = \left| \sum_{j=1}^{n-1} U_i^r(y_1,..,y_{j},x_{j+1},...,x_n) - U_i^r(y_1,..,y_{j+1},x_{j+2},...,x_n) \right| \\
 & \leq \sum_{j=1}^{n-1} \left| U_i^r(y_1,..,y_{j},x_{j+1},...,x_n) - U_i^r(y_1,..,y_{j+1},x_{j+2},...,x_n) \right| \\
 & \leq \sum_{j=1}^n \|x_i - y_i\|_1 \\
 & = \|x - y\|_1
\end{split}
\end{equation}
\end{proof}

Since Lipschitz continuity is maintained over maxima, with the same Lipschitz constant, we get the following result that says that best responses to nearby mixed strategy profiles are in fact approximate best responses.

\begin{lemma} \label{lemma:multiplayer-lipschitz}
$E_i : \Delta(A)_{-i} \rightarrow [0,1]$ is 1-Lipschitz when its domain is endowed with the $\ell_1$ norm. In particular, if $x,y \in \Delta(A)_{-i}$ and $\|x - y\|_1 \leq \frac{\epsilon}{2}$, then any $r \in BR^i_s(y)$ is an $\epsilon$ best response to $x$ for player $i$.
\end{lemma}

\subsection{Nash's Theorem with Discrete Approximations in Multiplayer Games}

Just as in bimatrix games, we have a notion of best response sets.

\begin{definition}[Multiplayer Best Response Sets] \label{def:multiplayer-BR-sets}
Let $G$ be a game with $n$ players, each with $k$ strategies. For a player $i$ and pure strategy $r \in A_i$, we define $P_i^j = \{x \in \Delta(A)_{-i} \ | \ BR^i_s(x) = r\}$. We say that $P_i^j$ is a best response set corresponding to strategy $r$ for player $i$, and note that $\{P_i^r\}_{r \in A_i}$ cover $\Delta(A)_{-i}$. 
\end{definition}

In bimatrix games our goal was to learn $\epsilon$-close labellings of the polytope partitions induced by best response sets. In the multiplayer setting that goal can be generalised.

\begin{definition}[Multiplayer $\epsilon$-close labellings]

Let $G$ be a game with $n$ players, each with $k$ actions, and let $i$ be a specific player in the game. Suppose that for each action $r \in A_i$, $\hat{P}_i^r \subseteq P_i^r$ is a closed set. We say the collection $\{\hat{P}_i^r\}_{r \in A_i}$ is an empirical labelling of $\Delta(A)_{-i}$. If in addition $\bigcup_{r \in A_i} \hat{P}_i^r$ is an $\epsilon$-net of $\Delta(A)_{-i}$ in the $\ell_1$ norm, we say that the collection $\{\hat{P}_i^r\}_{r \in A_i}$ is an $\epsilon$-close labelling of $\Delta(A)_{-i}$.
\end{definition}

As we will see shortly, if we manage to compute an $\frac{\epsilon}{2}$-close labelling for all $\Delta(A)_{-i}$, we have enough information to compute an approximate equilibrium. 

In bimatrix games, each $P_i^r$ is a polytope, but in multiplayer games, expected utilities are no longer linear in $\Delta(A)_{-i}$. Consequently, best response sets are semi-algebraic sets instead. This means that in general best response set are not connected and thus not convex. Without convexity and polytope structure we can no longer use binary search methods to learn $\epsilon$-close labellings. As mentioned before, this does not preclude us from computing an $\epsilon$-close labelling via a brute force method of querying an entire $\epsilon$-net of $\Delta(A)_{-i}$. Before we show this suffices however, we show the key result of this section: computing $\frac{\epsilon}{2}$-close labellings for all $\Delta(A)_{-i}$ suffices to compute $\epsilon$-WSNE. To do so we revisit Voronoi best response sets. In what follows we let $d(x,S)$ denote the infimum distance of a point, $x$ to a set $S$ in the $\ell_1$ norm.

\begin{definition}[Multiplayer Voronoi Best Response Functions and Best Response Sets]
Let $G$ be a game with $n$ players, each with $k$ actions, and let $i$ be a specific player in the game. Suppose that $\{\hat{P}_i^r\}_{r \in A_i}$ is an empirical labelling of $\Delta(A)_{-i}$. Player $i$'s Voronoi Best Response function is denoted by $V^i: \Delta(A)_{-i} \rightarrow A_i$. The function is defined as $V^i(x) = \argmin_{r \in A_i} d(x,\hat{P}_i^r)$. We also define player $i$'s Voronoi Best Response Sets as $V_i^r = \{x \in \Delta(A)_{-i} \ | \ V^i(x) = r \}$.
\end{definition}

If we invoke Lemma \ref{lemma:multiplayer-lipschitz}, we obtain the same result as in bimatrix games whereby Voronoi best responses are actually approximate best responses.

\begin{lemma}\label{lemma:multiplayer-voronoi-approximateBR}
Suppose that $\{\hat{P}_i^r\}_{r \in A_i}$ is an $\frac{\epsilon}{2}$-close labelling of $\Delta(A)_{-i}$, then Voronoi Best Response for player $i$ are $\epsilon$ Best Responses in $G$.
\end{lemma}

In addition, our definition of empirical labellings stipulated that $\hat{P}_i^r$ are all closed sets. This is a property which is inherited by Voronoi Best Response Sets. 

\begin{lemma}
Voronoi best response sets are closed.
\end{lemma}

\begin{proof}
The proof is identical to Lemma \ref{lemma:voronoi-closed} since $\ell_1$ distance is still a continuous function of the relevant domain $\Delta(A)_{-i}$. 
\end{proof}

We now define the generalisation to the Voronoi Best Response Correspondence from before.

\begin{definition}[Multiplayer Voronoi Best Response Correspondence]

Suppose that $x \in \Delta(A)$ is a mixed strategy profile. We define the Voronoi Best Response Correspondence $B^*: \Delta(A) \rightarrow \mathcal{P}( \Delta(A))$ as follows:
$$
B^*(x) = \prod_{i=1}^n conv(V^i(x_{-i})) \subseteq \Delta(A)
$$
\end{definition}

\begin{theorem}\label{thm:discrete-nash-multiplayer}
$B^*$ satisfies all the conditions of Kakutani's fixed point Theorem, and hence there exists a mixed strategy profile $x^* \in \Delta(A)$ such that $x^* \in B^*(x^*)$. In particular, if the Voronoi best response for $B^*$ arise from $\frac{\epsilon}{2}$-close labellings for all $\Delta(A)_{-i}$, then this in turn implies that $x^*$ is an $\epsilon$-WSNE.
\end{theorem}

\begin{proof}
As in the proof of Theorem \ref{thm:discrete-Nash}, the first three conditions of Kakutani's fixed point theorem are trivial. We focus on proving that $B^*$ has a closed graph. It suffices to show the following for any player $i$: if $\{x_n\}$ is a sequence in $\Delta(A)_{-i}$ that converges to $x$, and $\{y_n\}$ is a sequence in $\Delta(A)_i$ converging to $y$ with the property that $y_n \in conv(V^i(x_n))$ for all $n$, then $y \in conv(V^i(x))$.

To show this, we prove that there exists a constant $\delta > 0$ with the property that $B^1_\delta(x) \cap V_i^r$ if and only if $r \in V_i(x)$. Here $B^1_\delta(x)$ denotes the $\ell_1$ ball of radius $\delta$ around $x$. We prove this claim by contradiction.

Suppose instead that for every $\delta > 0$, $B^1_\delta(x)$ contains some point from a $V_i^r$, where $r \notin V^i(x)$. Let $\delta_1 > 0 $ be arbitrary, and let $z_1$ be the guaranteed point in $B^1_{\delta_1}(x)$ that is contained in a collection of $V_i^r$ where none belong to $V^i(x)$. Let $\delta_2 = \|x - z_1\|_1$. We continue in this fashion where for a given $\delta_k > 0$, we let $z_k \in B^1_{\delta_k}(x)$ be a point contained in a collection of $V_i^r$ where none belong to $V^i(x)$. Accordingly we define $\delta_{k+1} = \|x - z_k\|_1$. Since we can always continue this process, we recover a sequence $\{z_n\}$ of elements in $\Delta(A)_{-i}$ that converges to $x$. There are only a finite number of Voronoi Best Response sets, hence this sequence must contain an infinite subsequence of points belonging to the same voronoi best response set, say $V_i^{r'}$, where $r'$ does not belong to $V^i(x)$. This however implies that $x$ is a limit point of $V_i^{r'}$, and since this set is closed, this implies that $x \in V_i^{r'}$, which contradicts the fact that $r' \notin V^i(x)$, thus proving our desired claim.

Returning to the sequences $\{x_n\}$ and $\{y_n\}$, the existence of a fixed $\delta > 0$ with the property that $B^1_\delta(x) \cap V_i^r$ if and only if $r \in V_i(x)$ means that for some $N > 0$, if $n > N$, $x_n \in B^1_\delta(x)$, which in turn means that $y_n \in Conv(V^i(x))$. This in turn means that $y \in Conv(V^i(x))$, as desired. Applying this result to each $i$ yields the graph-closedness of $B^*$, thus establishing that $B^*$ satisfies all the properties of Kakutani's fixed point theorem. 

As for the second part of the theorem, suppose that $x^*$ is the guaranteed fixed point of $B^*$ as guaranteed by Kakutani's fixed point theorem. From Lemma \ref{lemma:multiplayer-voronoi-approximateBR}, we know that Voronoi Best Responses are $\epsilon$-Best responses. By the definition of $B^*$, the support of each $x^*_i \in \Delta(A)_i$ consists of $\epsilon$ best responses to $x^*$, thus establishing the fact that $x^*$ is an $\epsilon$-WSNE. 
\end{proof}

With the previous theorem in hand, we have established that computing a $\frac{\epsilon}{2}$-close labellings of all $\Delta(A)_{-i}$ suffices to compute an $\epsilon$-WSNE. Although the lack of convexity in best response sets prevents us from using binary search techniques, we can always query an $\epsilon$-net of $\Delta(A)_{-i}$ in the $\ell_1$ norm. In this vein, we construct an explicit $\epsilon$-net for $\Delta(A)_{-i}$.

\begin{lemma}
$M^n_\epsilon =  \left( \frac{2\epsilon}{n}\mathbb{Z} \right)^{n} \bigcap \Delta^n$ is an $\epsilon$-net in the $\ell_1$ norm for $\Delta^n$. Furthermore, $|M^n_\epsilon| = O\left(   (\frac{n}{2\epsilon} + n)^n \right)$.
\end{lemma}

\begin{proof}
The first claim follows from noting that lattice points lie on vertices of axis-aligned hypercubes of side-length $\frac{2\epsilon}{n}$. These cubes have a diagonal of length $2\epsilon$ in the $\ell_1$ norm, hence their centres are at most $\epsilon$-away from a given queried vertex.

As for the cardinality, using a stars and bars argument, one can see that if $S = \frac{1}{\kappa} \mathbb{Z}^{n} \cap \Delta^n$ for some $\kappa \in \mathbb{N}$, then $|S| = \binom {\kappa + n}{n} = O\left( (\kappa + n)^{n} \right)$. 
\end{proof}

Since $\Delta(A)_{-i}$ is a product of simplices, we can take products of the above $\epsilon$-net constructions to in turn obtain an $\epsilon$-net for $\Delta(A)_{-i}$.

\begin{lemma}\label{lemma-lattice-multiplayer}
Suppose that $\epsilon > 0$ and let $\epsilon' = \frac{\epsilon}{n-1}$. Furthermore, let us define the set $H_\epsilon^{n,k} = (M_{\epsilon'}^{k-1})^{n-1} \subset \Delta(A)_{-i} \cong (\Delta^{k-1})^{n-1}$. Then $H_\epsilon^{n,k}$ is an $\epsilon$ net for $\Delta(A)_{-i}$ in the $\ell_1$ norm. Furthermore $|H_\epsilon^{n,k}| = O\left(   (\frac{nk}{2\epsilon})^{nk} \right)$.
\end{lemma}

Querying all points in an $\epsilon$-net trivially gives rise to an $\epsilon$-close labelling. Consequently, with Lemma \ref{lemma-lattice-multiplayer} and Theorem \ref{thm:discrete-nash-multiplayer} in hand, we have proven our main result regarding the query complexity of computing an $\epsilon$-WSNE using adversarial Best Response Queries. 

\begin{theorem}
Suppose that $G$ is a game with $n$ players with $k$ pure strategies each. One can compute an $\epsilon$-WSNE of $G$ using $O\left(   n(\frac{nk}{\epsilon})^{nk} \right)$ adversarial Best Response Queries.
\end{theorem}

\section{Conclusion and Future Directions}\label{sec:conclusion}
In this paper we introduced the concept of learning $\epsilon$-close labellings of $(m,n)$-polytope partitions with membership queries, and derived query efficient algorithms for when either the dimension of the ambient simplex in the polytope partition, $m$, is held constant, or when the number of polytopes in the partition, $n$, is held constant.

Most importantly, we introduced a novel reduction from computing $\epsilon$-WSNE with best response queries to this geometric problem, thus allowing us to show that in the best response query model, computing $\epsilon$-WSNE of a bimatrix game has a finite query complexity. More specifically, for $m \times n$ games with $\min(m,n)$ constant, the query complexity is polynomial in $\max(m,n)$ and $\log \left( \frac{1}{\epsilon} \right)$. Furthermore, we partially extended our results from bimatrix games to multi-player games. Although the underlying geometry in multi-player games prevents us from using our results from learning polytope partitions, we were still able to show that querying a fine enough $\epsilon$-net of the mixed strategy space of all players suffices to compute an $\epsilon$-WSNE. 

As mentioned in the introduction, this geometric framework could be of use in other areas where Lipschitz continuous structures appear over domains with convex partitions. Upon further inspection, it is not difficult to see that polytope partitions do not need to be contained in $\Delta^m$, and in fact our algorithms extend to arbitrary ambient polytopes. 
Furthermore, it would be of great interest to create algorithms with a better query cost, prove lower bounds with regards to computing $\epsilon$-close labellings, or simply explore weaker query paradigms, such as noisy membership oracles. Finally, we have mentioned that in the multiplayer setting, best response sets are no longer polytopes, but rather semi-algebraic sets. It would be of interest to create learning algorithms for $\epsilon$-close labellings of these more complicated geometric objects, since doing so suffices to compute $\epsilon$-WSNE. 

\bibliographystyle{plainnat}

\bibliography{./refs.bib}

\end{document}